\documentclass{elsarticle}

\journal{Science of Computer Programming}

\usepackage{etex}
\usepackage[utf8]{inputenc}
\usepackage{textcomp}
\usepackage{amsthm}
\usepackage{amssymb}
\usepackage{amsmath}
\usepackage{paralist}
\usepackage{color}
\usepackage{comment}
\usepackage{tikz}
\usepackage[all]{xy}
\usetikzlibrary{matrix}
\usepackage{verbatim}
\usepackage{comment}
\usepackage{url}

\newif\ifdraft\drafttrue

\newcommand{\dedrule}[2]{\frac{#1}{#2}}
\newcommand{\trans}[1][]{\xrightarrow{\, {#1} \, }}

\newcommand{\Ran}[1]{\mathsf{Ran}_{#1}}
\newcommand{\Fam}[1]{\mathbb{#1}}
\newcommand{\coalg}[1]{\mathsf{coalg}(#1)}
\newcommand{\DL}[1]{\mathsf{DL}(#1)}
\newcommand{\counit}{\epsilon}

\newcommand{\tuple}[1]{\langle #1 \rangle}

\newcommand{\relR}{\mathrel{\mathcal{R}}}
\newcommand{\relS}{\mathrel{\mathcal{S}}}
\newcommand{\bisim}{\sim}
\newcommand{\obisim}{\sim_o}

\newcommand{\Var}{\mathcal{V}}

\newcommand{\tsigma}{T_{\Sigma}}
\newcommand{\openTerms}[1][]{T\Var}
\newcommand{\closedTerms}[1][]{T\emptyset}

\newcommand{\ca}{\textsf{a}}
\newcommand{\cb}{\textsf{b}}
\newcommand{\cc}{\textsf{c}}

\newcommand{\alt}{\textsf{alt}}

\newcommand{\head}{\textsf{hd}}
\newcommand{\tail}{\textsf{tl}}

\newcommand{\set}{\textit{Set}}

\newcommand{\hd}{\textsf{hd}}
\newcommand{\tl}{\textsf{tl}}

\newcommand{\Id}{\textit{Id}}
\newcommand{\id}{\textit{id}}
\newcommand{\cs}{\textit{cs}}
\newcommand{\st}{\textit{st}}

\newcommand{\proj}{\textit{proj}}
\newcommand{\ev}{\textit{ev}}
\newcommand{\buffer}{\rhd}

\newcommand{\baroplus}{\mathop{\bar \oplus}}

\newcommand{\nattrans}{\Longrightarrow}

\newcommand{\R}{\mathbb{R}}
\newcommand{\N}{\mathbb{N}}

\newcommand{\sem}[1]{|\![{#1}]\!|}
\newcommand{\osem}[1]{\sem{#1}_o}
\newcommand{\csem}[1]{\sem{#1}_c}
\newcommand{\asem}[1]{\sem{#1}_a}

\newcommand{\subs}{\varsigma}

\newcommand{\vx}{\mathcal{X}}
\newcommand{\vw}{\mathcal{W}}
\newcommand{\vy}{\mathcal{Y}}
\newcommand{\vz}{\mathcal{Z}}

\newcommand{\tgamma}{\tilde{\Gamma}}

\renewcommand{\b}{\mathbf{b}}
\newcommand{\f}{\mathbf{f}}

\newcommand{\tplus}{T_{\Sigma \cup \tilde \Sigma}\emptyset}

\newcommand{\cfunc}{\mathfrak{c}}

\newtheorem{theorem}{Theorem}[section]
\newtheorem{lemma}[theorem]{Lemma}
\newtheorem{corollary}[theorem]{Corollary}
\newtheorem{proposition}[theorem]{Proposition}
\theoremstyle{definition}
\newtheorem{example}[theorem]{Example}

\newtheorem{definition}[theorem]{Definition}

\begin{document}

\begin{frontmatter}

\title{Bisimilarity of Open Terms in Stream GSOS\tnoteref{mytitlenote}}

\tnotetext[mytitlenote]{
  The research leading to these results received funding from
    the European Research Council (FP7/2007-2013, grant agreement nr.~320571; as well as
	from the LABEX MILYON (ANR-10-LABX-0070, ANR-11-IDEX-0007), the
    project PACE (ANR-12IS02001) and the project REPAS (ANR-16-CE25-0011).
    \copyright 2018. This manuscript version is made available under the CC-BY-NC-ND 4.0 license \url{http://creativecommons.org/licenses/by-nc-nd/4.0/}. 
    This manuscript is published in Science of Computer Programming; the official version
    can be found at
    \url{https://doi.org/10.1016/j.scico.2018.10.007}}

%

\author[pis]{Filippo Bonchi}
\ead{filippo.bonchi@unipi.it}
\author[run]{Tom van Bussel}
\author[]{Matias David Lee}
\author[run]{Jurriaan Rot}
\ead{jrot@cs.ru.nl}


\address[run]{ICIS -- Radboud University, Nijmegen} 
\address[pis]{Department of Computer Science, University of Pisa}


\begin{abstract}
Stream GSOS is a specification format for operations and calculi on infinite sequences. 
The notion of bisimilarity provides a canonical
proof technique for equivalence of closed terms in such specifications. In this paper,
we focus on \emph{open terms}, which may contain variables, and which are equivalent 
whenever they denote the same stream for every possible instantiation of the variables. 
Our main contribution is to capture equivalence of open terms as bisimilarity
on certain Mealy machines, providing a concrete proof technique. Moreover, we introduce 
an enhancement of this technique, called bisimulation up-to substitutions, and show
how to combine it with other up-to techniques to obtain a powerful method
for proving equivalence of open terms. 
\end{abstract}

\begin{keyword}
bisimilarity \sep streams \sep open terms \sep GSOS \sep operational semantics
\end{keyword}

\end{frontmatter}

\section{Introduction}

Structural operational semantics (SOS) can be considered the de facto standard to define programming languages and process calculi.  
The SOS framework relies on defining a \emph{specification} consisting of a set of operation symbols, a set of labels or actions
and a set of inference rules. 
The inference rules describe the behaviour of each operation, typically depending on the behaviour 
of the parameters. The semantics is then defined in terms of a labelled transition 
system over  \emph{(closed) terms} constructed from the operation symbols. \emph{Bisimilarity}
of closed terms ($\bisim$) provides a canonical notion of behavioural equivalence. 

It is also interesting to study equivalence of \emph{open terms}, for instance to express properties of program constructors, like the commutativity of a non-deterministic choice operator. The latter can be formalised as the equation $\vx + \vy = \vy + \vx$, where the left and right hand sides are terms with variables $\vx,\vy$.
Equivalence of open terms ($\obisim$) is usually based on $\bisim:$ for all open terms $t_1,t_2$
\begin{equation}\label{eq:intro}t_1\obisim t_2 \text{ iff for all closed substitutions } \phi, \; \; \phi(t_1)\bisim \phi(t_2)\text{.}\end{equation}
The main problem of such a definition is the quantification over all substitutions: one would like to have an alternative characterisation, possibly amenable to the coinduction proof principle.
This issue has been investigated in several works, like~\cite{deSimone85,Rensink00,DBLP:journals/tcs/BaldanBB07,ACI12,LGCR09,PG10,ZE11}.

\begin{figure}
\setlength\arraycolsep{5pt}
$$\begin{array}{ccccc}
(a)&\dedrule{}{n \trans[n] 0} & \dedrule{x \trans[n] x' \quad y \trans[m] y'}{ x \oplus y \trans[n+m] x'\oplus y'} &
 \dedrule{x \trans[n] x' \quad y \trans[m] y'}{x \otimes y \trans[n \times m] (n \otimes y') \oplus (x' \otimes y)} &
\\[3ex]
\hline
\\
(b)&\dedrule{}{n \trans[n] 0}
& \dedrule{x \trans[n] x' \quad y \trans[m] y'}{ x \oplus y \trans[ n+m] x'\oplus y'} 
 & 
 \dedrule{x \trans[n] x' \quad y \trans[ m] y'}{x \otimes y \trans[ n \times m] (n \otimes y') \oplus (x' \otimes m.y')} &
 \dedrule{x \trans[ m] x'}{n.x \trans[n] m.x'}
\\[3ex]
\hline
\\
(c)& \dedrule{}{n \trans[b |n] 0}
& \dedrule{x \trans[b |n] x' \quad y \trans[b |m] y'}{ x \oplus y \trans[b | n+m] x'\oplus y'} 
 &
 \dedrule{x \trans[b |n] x' \quad y \trans[b | m] y'}{x \otimes y \trans[b | n \times m] (n \otimes y') \oplus (x' \otimes m.y')} &
 \dedrule{x \trans[b | m] x'}{n.x \trans[b | n] m.x'}
\\[3ex]
\hline
\\ (d)& 
\dedrule{}{\vx \trans[\subs | \subs(\vx)] \vx}
&
\dedrule{}{n \trans[\subs  |n] 0} \quad \ldots
&
 \dedrule{x \trans[\subs |n] x' \quad y \trans[\subs  | m] y'}{x \otimes y \trans[\subs  | n \times m] (n \otimes y') \oplus (x' \otimes m.y')} &
 \dedrule{x \trans[\subs | m] x'}{n.x \trans[\subs  | n] m.x'}
\end{array}$$
\caption{A stream GSOS specification (a) is transformed first into a monadic specification (b), then in a Mealy specification (c) and finally in a specification for open terms (d). In these rules, $n$ and $m$ range over real numbers, $b$ over an arbitrary set $B$, $\vx$ over variables and $\subs$ over substitutions of variables into  reals.}\label{fig:rulestreamcalculus}
\end{figure}

\medskip

In this paper, we continue this line of research, focusing on the simpler setting of \emph{streams}, which are  infinite sequences over a fixed data type. 
More precisely, we consider stream languages specified in the  \emph{stream GSOS format}~\cite{Klin11}, a syntactic rule format enforcing several interesting properties. 
 We show how to transform a stream specification into a \emph{Mealy machine} specification that defines the operational semantics of open terms. Moreover, a notion of bisimulation -- arising in a canonical way from the theory of coalgebras~\cite{Rutten00} -- exactly characterises $\obisim$ as defined in \eqref{eq:intro}.

Our approach can be illustrated by taking as running example the fragment of the stream calculus \cite{Rutten05} presented in Figure~\ref{fig:rulestreamcalculus}(a). The first step is to transform a stream GSOS specification (Section~\ref{sec:preliminaries}) into a \emph{monadic} one (Section~\ref{sec:from_full_to_monadic}). In this variant of GSOS specifications, no variable in the source of the conclusion appears in the target of the conclusion. For example, in the stream specification in Figure~\ref{fig:rulestreamcalculus}(a), the rule associated to ${\otimes}$ is not monadic. The corresponding monadic specification is illustrated in Figure~\ref{fig:rulestreamcalculus}(b). Notice this process requires the inclusion of a family of prefix operators (on the right of Figure~\ref{fig:rulestreamcalculus}(b)) that satisfy the imposed restriction. 

The second step -- based on~\cite{HK11} -- is to compute the \emph{pointwise extension} of the obtained specification (Section~\ref{sec:pe_monadic_gsos}). Intuitively,
we transform a specification of streams with outputs in a set $A$ into a specification of Mealy machines with inputs in an arbitrary set $B$ and outputs in $A$, by replacing each transition $ \trans[a] $  (for $a\in A$) with a transition $ \trans[b|a] $ for each input $b\in B$. See Figure~\ref{fig:rulestreamcalculus}(c). 

In the last step (Section \ref{sec:mm_over_ot}), we fix $B = \Var \to A$, the set of functions assigning outputs values in $A$ to variables in $\Var$.  To get the semantics of open terms, it only remains to specify the behaviours of variables in $\Var$. This is done with the leftmost rule in Figure~\ref{fig:rulestreamcalculus}(d).

As a result of this process, we obtain a notion of bisimilarity over open terms, which coincides
with behavioural equivalence of all closed instances, and provides a concrete
proof technique for equivalence of open terms. By relating open terms rather than all its possible instances, this novel technique often enables to use \emph{finite} relations, while standard bisimulation techniques usually require relations of infinite size on closed terms.
%
%
In Section~\ref{sec:compatibility}
we further enhance this novel proof technique by studying \emph{bisimulation up-to}~\cite{PS11}.
We combine several standard up-to techniques with the notion of \emph{bisimulation up-to substitutions}. 
These up-to techniques are useful to obtain small relations, and thereby simplify the bisimulation proof technique.
For instance, as we show in Section~\ref{sec:compatibility}, a rather intricate combination of up-to techniques
allows us to prove distributivity of shuffle product over sum in stream calculus,
that is, $x \otimes (y \oplus z) = (x \otimes y) \oplus (x \otimes z)$. 

Throughout the paper, we exhibit our approach through several examples, such as basic identities in the stream calculus.
More generally, since we take stream GSOS as the starting point, our approach is applicable precisely to all \emph{causal} functions
on streams: those functions where the $n$-th element of the output stream only depends on the first $n$ elements of
input. This follows from the known result that functions on streams are causal iff they are definable in stream GSOS~\cite{HansenKR16}.
For instance, the operations of the stream calculus of Rutten~\cite{Rutten05} are causal.
A typical example of a function which is one which drops every second element of the input stream. Such
functions fall outside the scope of the current approach. 

An earlier version of this work was presented at FSEN 2017.
The current paper extends the proceedings version~\cite{BonchiLR17}
with proofs (none of which could be included in the proceedings version),
and a new section (Section~\ref{sec:familiar}) which establishes a connection with very recent work on the \emph{companion}
of a functor~\cite{PousR17,BasoldPR17}. The companion provides a promising foundation 
for a general presentation of open terms in \emph{abstract GSOS}~\cite{TP97,Klin11}
 (a general format for operations on coalgebras formulated in terms of distributive laws, 
 of which stream GSOS is a special case). The aim of the last section is to establish a precise
link between the current approach (for streams) and the abstract
theory of distributive laws in~\cite{PousR17,BasoldPR17}, thus placing the current
work in a wider context. In particular, we find that 
the semantics of open terms GSOS introduced in this paper coincides with
the semantics that arises canonically from the companion. Further, we show
that the Mealy machine over open terms, which arises from the main construction in this paper,
can equivalently be obtained through a general correspondence between distributive laws and
coalgebras (over a functor category).
By its very nature, this section is more technical than
the previous sections, and assumes a deeper familiarity with category theory;
it is aimed at the reader with an interest in the more general
theory of coalgebra and distributive laws. 

\section{Preliminaries}
\label{sec:preliminaries}

We define the two basic models that form the focus of this paper: stream systems,
that generate infinite sequences (streams), and Mealy machines, that generate output streams
given input streams.

\begin{definition}
 A \emph{stream system} with outputs in a set $A$ is a pair $(X, \tuple{o,d})$ where 
 $X$ is a set of states and $\tuple{o,d}\colon X \to A \times X $ is a function, 
 which maps a state $x\in X$ to both an \emph{output value} $o(x)\in A$ and to a \emph{next state} $d(x) \in X$.
 We write $x \trans[a] y$ whenever $o(x)=a$ and $d(x) = y$.
\end{definition}

\begin{definition}
 A \emph{Mealy machine} with inputs in a set $B$ and outputs in a set $A$
 is a pair $(X, m)$ where $X$ is a set of states 
 and $m \colon X \to (A \times X)^B$ is a function assigning to each $x\in X$ a map $m(x) = \tuple{o_x, d_x}\colon B \to A \times X$.
 For all inputs $b \in B$, $o_x(b)\in A$ represents an output and $d_x(b) \in X$ a next state.
 We write $x \trans[b|a] y$ whenever $o_x(b)=a$ and $d_x(b) = y$.
\end{definition}

We recall the notion of \emph{bisimulation} for both models.

\begin{definition}
\label{def:bs:stream}
Let $(X,\tuple{o, d})$ be a stream system.
A relation ${\relR} \subseteq X\times X$ is a \emph{bisimulation} if for all $(x,y) \in {\relR}$,
$o(x) = o(y)$ and $(d(x),d(y)) \in {\relR}$. 
\end{definition}

\begin{definition}
\label{def:bs:mmachine}
Let $(X, m)$ be a Mealy machine.
A relation ${\relR} \subseteq X\times X$ is a \emph{bisimulation} if for all $(x,y) \in {\relR}$ and $b\in B$,
$o_x(b) = o_y(b)$ and $(d_x(b), d_y(b)) \in {\relR}$. 
\end{definition}

For both kind of systems, we say that $x$ and $y$ are \emph{bisimilar}, 
notation $x \bisim y$, if there is a bisimulation $\relR$ s.t. $x \relR y$.

Stream systems and Mealy machines, as well as the associated notions of bisimulation, are standard examples 
in the theory of \emph{coalgebras}~\cite{Rutten00}, a mathematical framework that allows to study state-based systems and their semantics at a high level of generality. In the current paper, the theory of coalgebras underlies and enables our main results.
 
Throughout this paper, we denote by $\set$ the category of sets and functions, 
by $\Id \colon \set \rightarrow \set$ the identity functor and by
$\id \colon X \rightarrow X$ the identity function on an object $X$.

\begin{definition}\label{def:coalgebra}
 Given a functor $F \colon \set \to \set$, an \emph{$F$-coalgebra} is a pair $(X,d)$,
 where $X$ is a set (called the carrier) and $d \colon X \to FX$ is a function (called the structure).
 An \emph{$F$-coalgebra morphism} from $d \colon X \to FX$ to $d'\colon Y \to FY$ is a map 
 $h \colon X \to Y$ such that $F h \circ d = d' \circ h$.
\end{definition}

Stream systems and Mealy machines are $F$-coalgebras for the functors
$FX = A \times X$ and $FX = (A \times X)^B$, respectively.

The semantics of systems modelled as coalgebras for a functor $F$ is provided
by the notion of \emph{final coalgebra}. 
A coalgebra $\zeta \colon Z \to FZ$ is called \emph{final} if
for every $F$-coalgebra $d \colon X \to FX$ there is a unique coalgebra morphism $\sem{-} \colon (X, d) \to (Z,\zeta)$.
We call $\sem{-}$ the \emph{coinductive extension} of $d$.
Intuitively, the set $Z$ of a final $F$-coalgebra $\zeta \colon Z \to FZ$ represents the universe of all possible $F$-behaviours and, for a coalgebra $(X,d)$, $\sem{-}\colon X \to Z$ represents its semantics: a function assigning a behaviour to all states in $X$. 
This motivates the following definition: given two states $x,y \in X$, 
$x$ and $y$ are said to be \emph{behaviourally equivalent} iff
$\sem{x} = \sem{y}$.
If $F$ preserves weak pullbacks then behavioural equivalence
coincides with bisimilarity, i.e., $x \bisim y$ iff $\sem{x} = \sem{y}$ (see~\cite{Rutten00}). 
This condition is satisfied by (the functors for) stream systems
and Mealy machines. In the sequel, by $\bisim$ we hence refer both to bisimilarity
and behavioural equivalence.

Final coalgebras for stream systems and Mealy machines will be pivotal for our exposition. We briefly recall them, following~\cite{Rutten00} and ~\cite{DBLP:journals/cuza/HansenR10}. 
The set $A^\omega$ of streams over $A$ carries a final coalgebra for the functor $FX=A\times X$. For every stream system 
$\tuple{o,d} \colon X \to A \times X$, the coinductive extension $\sem{-} \colon X \rightarrow A^\omega$  
assigns to a state $x \in X$ the stream $a_0a_1a_2 \dots$ whenever $x\trans[a_0]x_1\trans[a_1]x_2 \trans[a_2] \dots$

Recalling a final coalgebra for Mealy machines requires some more care. Given a stream $\beta\in B^\omega$, we write $\beta {\upharpoonright_n}$ for the prefix of $\beta$ of length $n$. A function $\cfunc \colon B^\omega \to  A^\omega$ is \emph{causal} if for all $n \in \N$ and all $\beta, \beta' \in B^\omega$:
$\beta {\upharpoonright_n} = \beta' {\upharpoonright_n}$ entails 
$\cfunc(\beta){\upharpoonright_n} = \cfunc(\beta'){\upharpoonright_n}$. Intuitively, $\cfunc$ is causal if the first $n$ symbols of the output stream depends only on the first $n$ symbols of the input stream.
The set $\Gamma(B^\omega,A^\omega) = \{\cfunc \colon B^\omega \to A^\omega \mid \cfunc \text{ is causal}\}$ carries a final coalgebra for the functor $FX=(A\times X)^B$. For every Mealy machine $m \colon X \to (A \times X)^B$, the coinductive extension $\sem{-} \colon X \rightarrow \Gamma(B^\omega,A^\omega) $  assigns to each state $x \in X$ and each input stream $b_0b_1b_2\dots \in B^\omega$ the output stream $a_0a_1a_2\dots \in A^\omega$ whenever $x\trans[b_0|a_0]x_1\trans[b_1|a_1]x_2 \trans[b_2|a_2] \dots$ The interested reader is referred to~\cite{DBLP:journals/cuza/HansenR10} for more details.

\subsection{System Specifications}\label{sec:spec}

Different kinds of transition systems, like stream systems or Mealy machines, can be specified by means of algebraic specification languages. The syntax is given by an \emph{algebraic signature} $\Sigma$, namely a collection of operation symbols $\{f_i \mid  i \in I\}$ 
where each operator $f_i$ has a (finite) arity $n_i \in \mathbb{N}$. For a set $X$, $T_\Sigma X$ denotes the set of $\Sigma$-terms with variables over $X$.
The set of closed $\Sigma$-terms is denoted by $T_{\Sigma}\emptyset$.
We omit the subscript when $\Sigma$ is clear from the context.

A standard way to define the operational semantics of these languages is by means of structural operational semantics (SOS) \cite{MRG07}. In this approach, the semantics of each of the operators is described by syntactic rules, 
and the behaviour of a composite system is given in terms of the behaviour of its components.
We recall \emph{stream GSOS}~\cite{Klin11,HansenKR16}, a specification format for stream systems.

\begin{definition}
\label{def:stream_gsos_rule}
A \emph{stream GSOS rule} $r$ for a signature $\Sigma$ and a set $A$ is a rule 
\begin{equation}
\label{eq:stream_gsos_rule}
\dedrule{x_1 \trans[a_1] x'_1 \quad \cdots \quad x_n \trans[a_n]x'_n}{f(x_1, \ldots, x_n) \trans[a] t}
\end{equation}
where 
$f \in \Sigma$ with arity $n$, 
$x_1, \ldots, x_n, x'_1, \ldots, x'_n $ are pairwise distinct variables, 
$t$ is a term built over variables $\{x_1, \ldots, x_n, x'_1, \ldots, x'_n\}$
and $a,a_1, \ldots, a_n \in A$.
We say that $r$ is triggered by $(a_1, \ldots, a_n) \in A^n$.

A \emph{stream GSOS specification} is a tuple $(\Sigma, A, R)$ where 
$\Sigma$ is a signature, $A$ is a set of actions and $R$ is a set of stream GSOS rules for $\Sigma$ and $A$ s.t. 
for each $f\in \Sigma$ of arity $n$ and each tuple $(a_1, \ldots, a_n)\in A^n$, 
there is only one rule $r \in R$ for $f$ that is triggered by $(a_1, \ldots, a_n)$.
\end{definition}

A stream GSOS specification allows us to extend any given stream system
$\tuple{o,d} \colon X \rightarrow A \times X$ to a stream
system $\overline{\tuple{o,d}} \colon TX \rightarrow A \times TX$, by induction:
the base case is given by $\tuple{o,d}$, and the inductive cases by the specification. 
This construction can be defined formally in terms of proof trees, or by coalgebraic 
means; we adopt the latter approach, which is recalled later in this section.

There are two important uses of the above construction:
(A) applying it to the (unique) stream system carried by the empty set $\emptyset$
yields a stream system over closed terms, i.e., of the form $\closedTerms \rightarrow A \times\closedTerms$;
(B) applying the construction to the final coalgebra yields
a stream system of the form $TA^\omega \rightarrow A \times TA^\omega$. 
The coinductive extension $\sem{-} \colon TA^\omega \rightarrow A^\omega$ of this stream system 
is, intuitively, the interpretation of the operations in $\Sigma$ on streams in $A^\omega$. 

\begin{figure}[th]
  \centering
  \makebox[0pt][c]{\parbox{1\textwidth}{%
      \begin{minipage}[b]{0.7\hsize}\centering
      \begin{gather*}
      \dedrule{}{\ca \trans[a] \ca}\ \forall a \in A \qquad 
      \dedrule{x \trans[a] x' \quad y \trans[b] y'}{\alt(x,y) \trans[a] \alt(y',x')}\ \forall a,b \in A
      \end{gather*}
      \caption{The GSOS-rules of our running example}
      \label{fig:rules}
    \end{minipage}
    \hfill
    \begin{minipage}[b]{0.285\hsize}\centering
 \begin{tikzpicture}
 \node (s1) at (0,0) {$\alt(\ca, \alt(\cb,\cc))$}; 
 \node (s2) at (0,-1) {$\alt(\alt(\cc,\cb),\ca)$};
  \path[->] (s1)  edge [out=200,in=160] node [left] {$a$} (s2);
  \path[->] (s2)  edge [out= 20,in=-20]node [right] {$c$} (s1);
      \end{tikzpicture}
      \caption{A stream system}
      \label{fig:ts}
    \end{minipage}
}}
\end{figure}

\begin{example}
\label{ex:running_example_specification}
Let $(\Sigma, A, R)$ be a stream GSOS specification where 
the signature $\Sigma$ consists of constants $\{\ca \mid a \in A\}$ and a binary operation $\alt$.
The set $R$ contains the rules in Figure~\ref{fig:rules}.  For an instance of (A), the term $\alt(\ca, \alt(\cb,\cc)) \in \closedTerms[\Sigma]$ defines the stream system depicted in Figure~\ref{fig:ts}. 
For an instance of (B),  the operation $\alt \colon  A^\omega \times A^\omega \to A^\omega$ maps streams $a_0a_1a_2\dots$, $b_0b_1b_2\dots $ to $a_0b_1a_2b_2 \dots$.
\end{example}

\begin{example}
 \label{ex:stream_calculus}
We now consider the specification $(\Sigma, \R, R)$ which is the fragment of the \emph{stream calculus} \cite{Rutten01,Rutten05} consisting of the constants $n \in \R$ and the binary operators \emph{sum} ${\oplus}$ and \emph{(convolution) product}~${\otimes}$. The set $R$ is defined in Figure~\ref{fig:rulestreamcalculus} (a).
 For an example of (A), consider $n \oplus m \trans[n+m] 0\oplus 0 \trans[0] 0\oplus 0 \trans[0] \dots$. For (B), the induced operation $\oplus \colon  \R^\omega \times \R^\omega \to \R^\omega$ is the pointwise sum of streams, i.e.,  it maps any two streams $n_0n_1\dots$, $m_0m_1\dots $ to $(n_0+m_0)(n_1+m_1) \dots$.
\end{example}

\begin{definition}
We say that a stream GSOS rule $r$ as in \eqref{eq:stream_gsos_rule} is \emph{monadic} 
if $t$ is a term built over variables $\{x'_1, \ldots, x'_n\}$.
A stream GSOS specification is \emph{monadic} if all its rules are monadic. 
\end{definition}

The specification of Example~\ref{ex:running_example_specification} satisfies the monadic stream GSOS format, while the one of Example~\ref{ex:stream_calculus} does not since, in the rules for $\otimes$, the variable $y$ occurs in the arriving state of the conclusion.

\medskip

The notions introduced above for stream GSOS, as well as the analogous ones for standard (labelled transition systems) GSOS \cite{BloomIM88}, can be reformulated in an abstract framework -- the so-called \emph{abstract GSOS}~\cite{TP97,Klin11} -- that will be pivotal for the proof of our main result.
In this setting, signatures are represented by (polynomial) functors, as follows. A signature $\Sigma$ presented
as a collection $\{f_i \mid  i \in I\}$ of operations (with arities $n_i \in \mathbb{N}$) 
corresponds to the following polynomial functor which, abusing notation, we also denote by $\Sigma$:
$$
\Sigma X = \coprod_{i \in I} X^{n_i} \qquad \Sigma(g)(f_i(x_1, \ldots, x_{n_i})) = f_i(g(x_1), \ldots, g(x_{n_i}))
$$
for every set $X$ and every map $g \colon X \rightarrow Y$ where,
in the definition of $\Sigma(g)$, we use $f_i(-)$ to refer to the $i$-th coproduct injection. 
For instance, the signature $\Sigma$ in Example~\ref{ex:running_example_specification} corresponds to the functor $\Sigma X = A + (X \times X)$, while the signature of Example \ref{ex:stream_calculus} corresponds to the functor $\Sigma X = \R + (X \times X) + (X \times X)$. Models of a signature are seen as algebras for the corresponding functor. 

\begin{definition}
 \label{def:algebra}
 Given a functor $F \colon  \set \to \set$, an \emph{$F$-algebra} is a pair $(X,d)$,
 where $X$ is the \emph{carrier set} and $d \colon  FX \to X$ is a function. An \emph{algebra homomorphism}
 from an $F$-algebra $(X,d)$ to an $F$-algebra $(Y,d')$ is a map
 $h \colon X \rightarrow Y$ such that $h \circ d = d' \circ Fh$. 
\end{definition}
Particularly interesting are \emph{initial algebras}: an $F$-algebra is called initial
if there exists a unique algebra homomorphism from it to every $F$-algebra.
For a functor corresponding to a signature $\Sigma$, the initial algebra is $(\closedTerms,\kappa)$ where $\kappa\colon  \Sigma \closedTerms \to \closedTerms$ maps, for each $i\in I$,  the tuple of closed terms $t_1,\dots t_{n_i}$ to the closed term $f_i(t_1,\dots t_{n_i})$. For every set $X$, we can define in a similar way $\kappa_X \colon \Sigma TX \to TX$. The \emph{free monad over $\Sigma$} consists of the endofunctor $T \colon  \set \to \set$, mapping every set $X$ to $TX$, together with the natural transformations
$\eta \colon  \Id\nattrans T$---interpretation of variables as terms---and 
$\mu \colon  TT \nattrans T$---glueing terms built of terms. (We use
the standard convention of denoting natural transformations with a double arrow $\nattrans$). 
Given an algebra $\sigma \colon  \Sigma Y \to Y$, for any
function $f \colon X \rightarrow Y$ there is a unique algebra
homomorphism $f^\dagger \colon TX \rightarrow Y$ from $(TX,\kappa_X)$ to $(Y,\sigma)$. 
%

\begin{definition} 
 \label{def:gsos_specification}
 An \emph{abstract GSOS specification} (of $\Sigma$ over $F$) is a natural transformation
 $\lambda \colon  \Sigma(\Id\times F) \nattrans FT$. 
 A \emph{monadic  abstract GSOS specification} (in short, monadic specification) is a natural transformation
 $\lambda \colon  \Sigma F \nattrans FT$.
\end{definition}
By instantiating the functor $F$ in the above definition to the functor for streams ($FX=A\times X$) one obtains all and only the stream GSOS specifications. Instead, by taking the functor for Mealy machines  ($FX=(A\times X)^B$) one obtains the Mealy GSOS format \cite{Klin11}: for the sake of brevity, we do not report the concrete definition here but this notion will be important in Section~\ref{sec:mm_over_ot} where, to deal with open terms, we transform stream specifications into Mealy GSOS specifications.

\begin{example}
For every set $X$, the rules in Example \ref{ex:running_example_specification} define a function $\lambda_X : A + (A\times X) \times (A\times X) \to (A\times \tsigma X)$ as follows: each $a\in A$ is mapped to $(a, \ca)$ and each pair $(a,x'),(b,y') \in (A\times X) \times (A\times X)$ is mapped to $(a, \alt(y',x'))$~\cite{Klin11}.
\end{example}

We focus on monadic distributive laws for most of the paper, and since they are slightly simpler
than abstract GSOS specifications, 
we only recall the relevant concepts for monadic distributive laws. However, we note that
 the concepts below can be extended to abstract GSOS specifications; see, e.g.,~\cite{Bartels04,Klin11} for details.

A monadic abstract GSOS specification induces a distributive law $\rho \colon T F \nattrans FT$,
which is a distributive law of the (free) monad $(T,\eta,\mu)$ over the functor $F$,
i.e., it makes the following diagram commute:
\begin{equation}
\label{eq:dl-monad}
\xymatrix{
	TTF \ar@{=>}[d]_{\mu} \ar@{=>}[r]^{T\rho}
		& TFT \ar@{=>}[r]^{\rho T}
		& FTT \ar@{=>}[d]^{F\mu} \\
	TF \ar@{=>}[rr]^{\rho} 
		& & FT  \\
	F \ar@{=>}[u]^{\eta F} \ar@{=>}[urr]_{F\eta}
}
\end{equation}
The distributive law $\rho$ 
allows us to extend any $F$-coalgebra $d \colon X \rightarrow FX$ to an $F$-coalgebra
on terms:

$$
\xymatrix{TX \ar[r]^{Td} & TFX \ar[r]^{\rho_X} & FTX}
$$
This construction generalises and formalises the aforementioned extension of stream systems to terms by means of a
stream GSOS specification. 
In particular, (A) the unique coalgebra on the empty set $!\colon  \emptyset \to F \emptyset$ yields an $F$-coalgebra
on closed terms $\closedTerms \to F\closedTerms$. If $F$ has a final coalgebra $(Z,\zeta)$, the unique (hence depicted by a dotted line) morphism $\sem{-}_c \colon  \closedTerms \to Z$ defines the \emph{semantics of closed terms}.
\begin{center}
\begin{tikzpicture}
  \matrix (m) [matrix of math nodes,row sep=0.3em,column sep=3em]
  {
     T\emptyset & T F \emptyset  & FT\emptyset \\
     & (A) & \\
     Z &  & FZ \\};
  \path[-stealth]
    (m-1-1) edge [dotted] node [left] {\scriptsize $\sem{-}_c $} (m-3-1)
            edge node [above] {\scriptsize $T !$} (m-1-2)
    (m-1-2) edge node [above] {\scriptsize $\rho_\emptyset$} (m-1-3)
    (m-3-1) edge node [below] {\scriptsize $\zeta$}  (m-3-3)
    (m-1-3) edge [dotted] node [right] {\scriptsize $F\sem{-}_c $} (m-3-3);
\end{tikzpicture}
\begin{tikzpicture}
  \matrix (m) [matrix of math nodes,row sep=0.3em,column sep=3em]
  {
     TZ & T F Z  & FTZ \\
     & (B) & \\
     Z &  & FZ \\};
  \path[-stealth]
    (m-1-1) edge [dotted] node [left] {\scriptsize $\asem{-} $} (m-3-1)
            edge node [above] {\scriptsize $T \zeta$} (m-1-2)
    (m-1-2) edge node [above] {\scriptsize $\rho_Z$} (m-1-3)
    (m-3-1) edge node [below] {\scriptsize $\zeta$}  (m-3-3)
    (m-1-3) edge [dotted] node [right] {\scriptsize $F\asem{-} $} (m-3-3);
\end{tikzpicture} 
\end{center}

Further (B), the final coalgebra $(Z,\zeta)$ yields
a coalgebra on $T Z$. By finality, we then
obtain a $T$-algebra over the final $F$-coalgebra, which we denote by $\asem{-} \colon T Z \to Z$ and we call it the \emph{abstract semantics}. We define the \emph{algebra induced by}  $\lambda$ as the $\Sigma$-algebra $\sigma \colon  \Sigma Z \to Z$ given by
\begin{equation}\label{eq:alg_induced_by_lambda}
\xymatrix@C=1.3cm{
	 \Sigma Z \ar[r]^{\Sigma \eta_Z}
	 	& \Sigma TZ \ar[r]^{\kappa_Z}
	 	& TZ \ar[r]^{\sem{-}_a} & Z
}.
\end{equation}

\section{Making arbitrary stream GSOS specifications monadic}\label{sec:from_full_to_monadic}

The results presented in the next section are restricted to monadic specifications, but one can prove them for arbitrary GSOS specifications by exploiting some auxiliary operators, introduced in~\cite{HK11} with the 
name of \emph{buffer}. 
Theorem~\ref{th:upto} in Section~\ref{sec:compatibility}
only holds for monadic GSOS specifications. This does not restrict the applicability of our approach: as we show below, arbitrary stream GSOS specifications can be turned into monadic ones.

Let  $(\Sigma, A, R)$ be a stream GSOS specification. 
The extended signature $\tilde \Sigma$ is given by $\{\tilde f \mid f \in \Sigma\} \cup \{a.\_ \mid a \in A\}$.
The set of rules $\tilde R$ is defined as follows:
\begin{itemize}
 \item For all $a, b \in A$, $\tilde R$ contains the following rule 
 
\begin{equation} \label{rule:prefix}
 \dedrule{x \trans[b] x'}{a.x \trans[a] b.x'}
 \end{equation}
 \item For each rule 
 $r=\dedrule{x_1 \trans[a_1] x'_1 \quad \cdots \quad x_n \trans[a_n]x'_n}%
            {f(x_1, \ldots, x_n) \trans[a] t(x_1, \ldots, x_n,x'_1, \ldots, x'_n)} \in R$, the set
 $\tilde R$ contains 
 \begin{equation}
  \label{rule:general_rule_trans}
 \tilde r= \dedrule{x_1 \trans[a_1] x'_1 \quad \cdots \quad x_n \trans[a_n]x'_n}
            {\tilde f(x_1, \ldots, x_n) \trans[a] \tilde t(a_1.x'_1, \ldots, a'_n.x'_n,x'_1, \ldots, x'_n)} 
 \end{equation} 
 where $\tilde t$ is the term obtained from $t$ by replacing each $g \in \Sigma$ 
 by $\tilde g \in \tilde \Sigma$.
\end{itemize}

\noindent
The specification $(\tilde \Sigma, A, \tilde R)$ is now monadic and preserves the original semantics as stated by the following result.
For a detailed proof, see~\ref{sec:the_enc_is_sound}.

\begin{theorem}
\label{th:translation}
Let $(\Sigma, A, R)$ be a stream GSOS specification and $(\tilde \Sigma, A, \tilde R)$ be the corresponding monadic one. Then, 
for all $t \in T_{\Sigma}\emptyset$, $t \bisim \tilde t$.
\end{theorem}

\begin{example}
 \label{ex:stream_calculus:monadic}
Consider the non-monadic specification in Example~\ref{ex:stream_calculus}. The corresponding monadic specification consists of the rules in Figure~\ref{fig:rulestreamcalculus} (b) where, to keep the notation light, we used operation symbols $f$ rather than $\tilde{f}$. 
\end{example}

In~\cite{BasoldPR17}, it is conjectured that a variant of Theorem~\ref{th:translation} goes through
at the more general level of abstract/monadic GSOS for arbitrary endofunctors (rather than only streams systems),
but a proof of such a general statement is still missing. 
In~\cite{PousR17}, it is shown more abstractly that, for polynomial functors on $\set$ (such as our stream functor), any operation expressible by a GSOS specification is also expressible by a plain distributive law (of one functor over the other),
but that proof is less constructive: it does not result in a concrete new specification, but only shows existence.

\section{Pointwise Extensions of Monadic GSOS Specifications} 
\label{sec:pe_monadic_gsos}

The first step to deal with the semantics of open terms induced by a stream GSOS specification is to transform the latter into a Mealy GSOS specification.
We follow the approach in~\cite{HK11} which is defined for arbitrary GSOS but, as motivated in Section \ref{sec:from_full_to_monadic}, we restrict our attention to  monadic specifications.

Let  $(\Sigma, A, R)$ be a monadic stream GSOS specification and $B$ some input alphabet.
The corresponding monadic Mealy GSOS specification is a tuple $(\Sigma, A, B, \overline{R})$, 
where $\overline{R}$ is the least set of Mealy rules which contains, 
for each stream GSOS rule $
  r=\dedrule{x_1 \trans[a_1] x'_1 \quad \cdots \quad x_n \trans[a_n]x'_n}%
            {f(x_1, \ldots, x_n) \trans[a] t(x'_1, \ldots, x'_n)}\in R$ and $b\in B$, 
            the Mealy rule $\overline r_b$ defined by 
 \begin{equation}
  \label{rule:toMealy}
 \overline r_b= \dedrule{x_1 \trans[b | a_1] x'_1 \quad \cdots \quad x_n \trans[b | a_n]x'_n}
            { f(x_1, \ldots, x_n) \trans[b | a]  t(x'_1, \ldots, x'_n)} 
 \end{equation} 
An example of this construction is shown in Figure~\ref{fig:rulestreamcalculus} (c).

Recall from Section \ref{sec:preliminaries} that any abstract GSOS specification induces a $\Sigma$-algebra on the final $F$-coalgebra. Let $\sigma \colon  \Sigma A^\omega \to A^\omega $ be the algebra induced by the stream specification and $\overline{\sigma} \colon  \Sigma \Gamma(B^\omega,A^\omega)  \to \Gamma(B^\omega,A^\omega) $ the one induced by the corresponding Mealy specification. Theorem~\ref{th:pointwise_ext}, later in this section, informs us that $\overline{\sigma}$ is the \emph{pointwise extension} of $\sigma$.

\begin{definition}\label{def:pw-ext-concrete}
Let  $g \colon  (A^\omega)^n \to A^\omega$ and $\bar{g} \colon (\Gamma(B^\omega,A^\omega))^n \to  \Gamma(B^\omega,A^\omega) $ be two functions.
We say that $\bar{g}$ is  the \emph{pointwise extension} of $g$ iff for all $\cfunc_1, \ldots ,\cfunc_{n} \in  \Gamma(B^\omega,A^\omega) $  and $\beta \in B^\omega$,
$\bar{g}(\cfunc_1, \ldots, \cfunc_{n})(\beta) = g (\cfunc_1(\beta), \ldots, \cfunc_{n}(\beta ))$.
This notion is lifted in the obvious way to $\Sigma$-algebras for an arbitrary signature $\Sigma$.
\end{definition}

\begin{example}
Recall the operation $\oplus \colon  A^\omega \times A^\omega \to A^\omega$ from Example~\ref{ex:stream_calculus} that arises from the specification in Figure~\ref{fig:rulestreamcalculus} (a) (it is easy to see that the same operation also arises from the monadic specification in Figure~\ref{fig:rulestreamcalculus} (b)). Its pointwise extension $\bar{\oplus} \colon  \Gamma(B^\omega,\R^\omega) \times \Gamma(B^\omega,\R^\omega) \to \Gamma(B^\omega,\R^\omega)$ is defined for all  $\cfunc_1,\cfunc_2 \in \Gamma(B^\omega, \R^\omega)$ and $\beta\in B^\omega$ as
$(\cfunc_1 \bar{\oplus} \cfunc_2)(\beta) = \cfunc_1(\beta) \oplus \cfunc_2(\beta)$.  
Theorem~\ref{th:pointwise_ext} tells us that $\bar{\oplus}$ arises from the corresponding Mealy GSOS specification (Figure~\ref{fig:rulestreamcalculus}(c)).
\end{example}

In \cite{HK11}, the construction in~\eqref{rule:toMealy} is generalised from stream specifications to arbitrary abstract GSOS. The key categorical tool is the notion of \emph{costrength} for an endofunctor $F\colon  \set \to \set$. Given two sets $B$ and $X$, we first define $\epsilon^b \colon  X^B \to X$ as $\epsilon^b(f) = f(b)$ for all $b\in B$. Then, $\cs^F_{B,X} \colon  F (X^B) \to (FX)^B$ is a natural map in $B$ and $X$, given by $\cs^F_{B,X}(t)(b) = (F\epsilon^b)(t)$.
See \ref{sec:strengthco} for some additional basic properties of costrength.

Now, given a monadic specification $\lambda \colon  \Sigma F \nattrans FT$, we define $\bar{\lambda} \colon  \Sigma(F^B) \nattrans (FT)^B$ as the natural transformation that is defined for all sets $X$ by
\begin{equation}\label{eq:lambdabar}
\xymatrix@C=1.3cm{
	 \Sigma(FX)^B \ar[r]^{\cs^\Sigma_{B,FX}}
	 	& (\Sigma FX)^B \ar[r]^{\lambda^B_X}
	 	& (FT X)^B
}.
\end{equation}
Observe that $\bar{\lambda}$ is also a monadic specification, but for the functor $F^B$ rather than the functor $F$. The reader can easily check that for $F$ being the stream functor $FX=A\times X$, the resulting $\bar{\lambda}$ is indeed the Mealy specification corresponding to $\lambda$ as defined in~\eqref{rule:toMealy}.

It is worth to note that the construction of $\bar{\lambda}$ for an arbitrary abstract GSOS $\lambda \colon  \Sigma (\Id\times  F) \nattrans FT$, rather than a monadic one, would not work as in~\eqref{eq:lambdabar}. The solution devised in~\cite{HK11} consists of introducing some auxiliary operators as already discussed in Section \ref{sec:from_full_to_monadic}. The following result has been proved in~\cite{HK11} for arbitrary abstract GSOS, with these auxiliary operators. Our formulation is restricted to monadic specifications. 
\begin{theorem}
\label{th:pointwise_ext}
Let $F$ be a functor with a final coalgebra $(Z, \zeta)$, 
and let $(\bar{Z}, \bar{\zeta})$ be a final $F^B$-coalgebra. 
Let $\lambda \colon  \Sigma F \nattrans FT$ be a monadic distributive law, and
$\sigma \colon  \Sigma Z \to Z$ the algebra induced by it. 
The algebra $\bar \sigma \colon  \Sigma \bar Z \to \bar Z$ induced by $\bar{\lambda}$ is a pointwise extension of $\sigma$.
\end{theorem}

In the theorem above, the notion of pointwise extension should be understood as a generalisation of Definition~\ref{def:pw-ext-concrete} to arbitrary final $F$ and $F^B$-coalgebras. This generalised notion, that has been introduced in~\cite{HK11}, will not play a role for our paper where $F$ is fixed to be the stream functor $FX=A\times X$. Nevertheless, for the sake of completeness, we report its definition in~\ref{app:pointwise_ext},
and prove Theorem~\ref{th:pointwise_ext} at this general level (\ref{sec:proof_of_pointwise_ext}).

\begin{example}
We show how to obtain, in the abstract setting, the corresponding Mealy GSOS specification from a stream GSOS specification.
To simplify the example we only consider the operator of ${\oplus}$ of our running example.
The syntax, i.e. $\oplus$, defines a functor $\Sigma X = X \times X$
and the behavior functor of stream systems is $F = \R \times X$.
The monadic GSOS specification, i.e. the family of rules for ${\oplus}$ in Figure~\ref{fig:rulestreamcalculus}~{(b)}, induces a natural transformation
$$
\lambda : (\R \times {-}) \times (\R \times {-}) \nattrans (\R \times T{-})
$$
whose $X$-component is given by: 
\begin{equation}
 \begin{tabular}{lccc}
 $\lambda_X :$ & $(\R \times X ) \times (\R \times X)$ & $\nattrans$ & $\R \times TX$\\[1ex]
 & $\tuple{\tuple{a, x_1}, \tuple{b, x_2}}$ & $\longmapsto$ & $\tuple{a+b, \tuple{x_1, x_2}}$ 
 \end{tabular}
\end{equation}
For defining the pointwise extension $\baroplus$ of $\oplus$, 
we have to consider the behavior functor $F^BX = (\R \times X)^B$.
Following the construction of $\bar \lambda$, an $X$-component of a $\bar \lambda$ is of the type:
$$
\bar \lambda_X : (\R \times X)^B \times (\R \times X)^B \nattrans (\R \times TX)^B
$$
Let $\phi(b) = \tuple{\phi_0(b), \phi_1(b)}$ for $\phi \in (\R \times X)^B$ and $b\in B$.
By instantiation of (\ref{eq:lambdabar}), $\bar \lambda_X$ is defined by:
\begin{center}
\begin{tikzpicture}
  \matrix (m) [matrix of math nodes,row sep=2em,column sep=2.5em]
  {
     \tuple{\phi, \psi} & \in  &  (\R \times X)^B \times (\R \times X)^B \\
     \lambda b.\tuple{\phi_0(b), \phi_1(b), \psi_0(b), \psi_1(b)} 
     & \in  
     & (\R \times X \times \R \times X)^B \\
     \lambda b.\tuple{\phi_0(b) + \psi_0(b), \tuple{\phi_1(b), \psi_1(b)}}
     & \in 
     & (\R \times TX)^B \\
   };     
   \path[-stealth]
     (m-1-1) edge (m-2-1)
     (m-2-1) edge (m-3-1)
     (m-1-3) edge node [right] {\scriptsize $\cs^\Sigma_{B, \R\times X}$} (m-2-3)
     (m-2-3) edge node [right] {\scriptsize $\lambda_X^B$} (m-3-3);
\end{tikzpicture}
\end{center}
Using the notation $\phi \trans[b|\phi_0(b)] \phi_1(b)$ to denote $\phi(b) = \tuple{\phi_0(b), \phi_1(b)}$
and taking into account the definition of $\bar \lambda$, we get what we want, 
i.e. the rules defining the semantics of  $\baroplus$ are the family of the family of Mealy GSOS rules for ${\oplus}$ in Figure~\ref{fig:rulestreamcalculus}~(c). Notice we have reused the operation symbols ${f}$ rather than ${\bar f}$ to keep the notation light.
\end{example}

\section{Mealy Machines over Open Terms}
\label{sec:mm_over_ot}

We now consider the problem of defining a semantics for the set of open terms $\openTerms$ for a fixed set of variables $\Var$.
Our approach is based on the results in the previous sections: we transform a monadic GSOS specification for streams with outputs in $A$ into a Mealy machine with inputs in $A^\Var$ and outputs in $A$, i.e.,  a coalgebra for the functor $FX=(A \times X)^{A^\Var}$. 
The coinductive extension of this Mealy machine provides the open semantics:
for 
each open term $t\in \openTerms$ and variable assignment $\psi \colon \Var \rightarrow A^\omega$, it gives an appropriate output stream in $A^{\omega}$.
This is computed in a stepwise manner: for an input $\subs \colon  \Var \to A$, representing ``one step'' of a variable assignment $\psi$, we obtain one step of the output stream.

We start by defining a Mealy machine $c \colon  \Var \to (A \times \Var)^{A^\Var}$ on the set of variables $\Var$ as
on the left below, for all $\vx \in \Var$ and $\subs \in A^\Var$:
\begin{equation}\label{eq:proj-mm}
\begin{gathered}
c(\vx)(\subs) = (\subs(\vx), \vx) 
\qquad \qquad 
\xymatrix{
	\vx \ar@(ru,dr)^{\subs \mid \subs(\vx)}
}
\end{gathered}
\end{equation}
Concretely, this machine has variables as states and for each $\subs \colon  \Var \to A$ a self-loop, as depicted on the right.
Now, let $\lambda \colon  \Sigma (A \times {-}) \Rightarrow A \times T$ be 
a monadic abstract stream specification and 
$\bar{\lambda} \colon  \Sigma ((A \times {-})^{A^\Var}) \Rightarrow (A \times T({-}))^{A^\Var}$
be the induced Mealy specification, as defined in~\eqref{eq:lambdabar}. 
As mentioned in Section~\ref{sec:preliminaries}, $\bar{\lambda}$
defines a distributive law 
$\rho \colon T((A \times {-})^{A^\Var}) \Rightarrow (A \times T({-}))^{A^\Var}$,
which allows to extend $c$ (see~\eqref{eq:proj-mm}) to a coalgebra
$m_\lambda \colon \openTerms \rightarrow (A \times \openTerms)^{A^\Var}$, given by
\begin{equation}\label{eq:mm-ot}
\xymatrix{
\openTerms \ar[r]^-{Tc} 
	& T(A \times \Var)^{A^\Var} \ar[r]^-{\rho_\Var}
	& (A \times \openTerms)^{A^\Var}\text{.}
}
\end{equation}
\noindent This is the Mealy machine of interest.
Intuitively, it is constructed by computing the pointwise extension
of the original stream GSOS specification (as in the previous section) and
then defining the Mealy machine by induction on terms. The base
case is given by~\eqref{eq:proj-mm}, and the inductive cases 
by the (pointwise extended) specification. We first give a few examples.

\begin{example}\label{ex:alt-ex-open}
Consider the stream specification $\lambda$ of the operation $\alt$, given in Example~\ref{ex:running_example_specification}.
The states of the Mealy machine $m_\lambda$ are the open terms $\openTerms$.
The transitions of terms are defined by the set of rules 
     $$ \dedrule{}{\ca \trans[\subs|a] \ca} \qquad
      \dedrule{x \trans[\subs|a] x' \quad y \trans[\subs|b] y'}{\alt(x,y) \trans[\subs|a] \alt(y',x')} \qquad 
      \text{ for all } \subs \colon  \Var \rightarrow A \text{ and } a,b \in A 
      $$
together with the transitions for the variables as in~\eqref{eq:proj-mm}. For instance, for each $\vx, \vy, \vz \in \Var$ and all $\subs,\subs' \colon  \Var \rightarrow A$, we have the following transitions
in $m_\lambda$:
\begin{center}
 \begin{tikzpicture}
 \node (s1) at (0,0) {$\alt(\vx, \alt(\vy,\vz))$}; 
 \node (s2) at (0,-1) {$\alt(\alt(\vz,\vy),\vx)$};
  \path[->] (s1)  edge [out=200,in=160] node [left] {$\subs | \subs(\vx)$} (s2);
  \path[->] (s2)  edge [out= 20,in=-20]node [right] {$\subs' | \subs'(\vz)$} (s1);
      \end{tikzpicture}
\end{center}
\end{example}

\begin{example}\label{ex:streamcalculus-open}
For the fragment of the stream calculus introduced in Example~\ref{ex:stream_calculus}, the Mealy machine over open terms is defined by the rules in Figure \ref{fig:rulestreamcalculus}(d). Below we draw the Mealy machines of some open terms that will be useful later. 
\begin{equation*}
\xymatrix{
	\vx \oplus \vy \ar@(ul,ur)^{\subs \mid \subs(\vx) + \subs(\vy)}
}
\quad
\xymatrix{
	\vy \oplus \vx \ar@(ul,ur)^{\subs \mid \subs(\vy) + \subs(\vx)}
}
\quad 
\quad 
\xymatrix{
	(\vx \oplus \vy) \oplus \vz \ar@(ul,ur)^{\subs \mid  (\subs(\vx) + \subs(\vy)) +  \subs(\vz)  }
}
\quad 
\xymatrix{
	\vx \oplus (\vy \oplus \vz) \ar@(ul,ur)^{\subs \mid  \subs(\vx) + (\subs(\vy) +  \subs(\vz))  }
}
\end{equation*}
\end{example}

We define the open semantics below
by the coinductive extension of $m_{\lambda}$. Let $\tgamma = \Gamma((A^\Var)^\omega, A^\omega)$
be the set of causal functions $\cfunc \colon  (A^\Var)^\omega \to A^\omega$, which is
the carrier of the final coalgebra for the functor $FX= (A\times X)^{A^\Var}$.
Notice that a function $\cfunc \colon (A^\Var)^\omega \rightarrow A^\omega$ can equivalently be presented as a function
$\tilde{\cfunc} \colon  (A^\omega)^\Var \to A^\omega$ (swapping the arguments in the domain). Given such a function $\cfunc \colon (A^\Var)^\omega \rightarrow A^\omega$ and a function $\psi \colon \Var \rightarrow A^\omega$, in the sequel,
we sometimes abuse of notation by writing $\cfunc(\psi)$ where we formally mean $\tilde{\cfunc}(\psi)$.

\begin{definition}\label{def:open}
	Let $\lambda \colon  \Sigma (A \times {-}) \Rightarrow A \times T$ be a monadic abstract stream GSOS specification.
	The \emph{open semantics} of $\lambda$ is the coinductive extension $\osem{-} \colon \openTerms \rightarrow \tgamma$
	of the Mealy machine $m_\lambda \colon \openTerms \rightarrow (A\times \openTerms)^{A^\Var}$ defined in~\eqref{eq:mm-ot}.
\end{definition}
Note that the open semantics $\osem{-}$ assigns to every open term $t$ a causal function on streams. Thus, two 
terms are behaviourally equivalent (identified by $\osem{-}$) precisely if they denote the same causal function. 

Since $\osem{-}$ is defined coinductively (as the unique morphism into the final coalgebra), behavioural equivalence of open terms can now be checked by means of bisimulations on Mealy machines (Definition \ref{def:bs:mmachine}). We define \emph{open bisimilarity}, denoted by $\obisim$, as the greatest bisimulation on $m_\lambda$. Obviously, for all open terms $t_1,t_2\in \openTerms$ it holds that $t_1 \obisim t_2$ iff $\osem{t_1}=\osem{t_2}$. The following result provides another useful characterisation of $\osem{-}$.

\begin{lemma}
 \label{lm:semantics}
 	Let $\lambda$ be a monadic abstract stream GSOS specification, with induced
 	algebra $\sigma \colon \Sigma A^\omega \rightarrow A^\omega$. 
 	Let $\bar{\lambda}$ be the corresponding Mealy specification, with induced algebra
 	 $\bar{\sigma} \colon \Sigma \tgamma \rightarrow \tgamma$.
 	Then the open semantics $\osem{-}$ is the unique homomorphism making the diagram below commute:

\begin{equation}
\begin{tabular}{c}
 \begin{tikzpicture}[every node/.style={text depth=0.5ex,text height=1.7ex}]
  \matrix (m) [matrix of math nodes,row sep=1.5em,column sep=3em]
  {
    \Sigma T \Var & \Sigma \tgamma   \\
    T\Var & \tgamma  \\
    \Var & \\
   };     
   \path[-stealth]
     (m-1-1) edge node [above] {\scriptsize $\Sigma \osem{\_ }$} (m-1-2)
             edge node [left] {\scriptsize $\kappa_{\Var}$} (m-2-1)
     (m-1-2) edge node [right] {\scriptsize $\bar \sigma$} (m-2-2)
     (m-2-1) edge node [above] {\scriptsize $\osem{\_ }$} (m-2-2)
     (m-3-1) edge node [left] {\scriptsize $\eta_{\Var}$} (m-2-1)
             edge node [below] {\scriptsize $\proj$} (m-2-2)     ;
\end{tikzpicture}
\end{tabular}
\end{equation} 

\noindent
where $\eta$ and $\kappa$ are defined by initiality (Section~\ref{sec:preliminaries}), 
and for each $\vx \in \Var$ and $\psi \colon  \Var \to A^\omega$, $\proj(\vx)(\psi) = \psi(\vx)$.
\end{lemma}
\begin{proof}
The open semantics is, by definition, the coinductive extension of $m_\lambda$, as
depicted on the right below:
$$
\xymatrix@C=2cm{
	\Var \ar[dd]_{c} \ar[r]^{\eta_\Var}&
\openTerms \ar[d]^-{Tc} \ar[r]^{\osem{-}}
& \tgamma \ar[dd]^{\zeta}
\\
	& \tsigma(A \times \Var)^{A^\Var} \ar[d]^-{\rho_\Var} & \\
	(A \times \Var)^{A^\Var} \ar[r]_{(\id \times \eta_{\Var})^{A^\Var}}
	& (A \times \openTerms)^{A^\Var} \ar[r]_{(\id \times \osem{-})^{A^\Var}}
	& (A \times \tgamma)^{A^\Var}
}
$$
where $(\tgamma, \zeta)$ is the final $(A \times -)^{A^\Var}$-coalgebra.
It is a standard fact in the theory of bialgebras that $\osem{-}$ is
an algebra morphism as in the statement of the lemma, so it remains to prove 
$\proj = \osem{-} \circ \eta$. To this end, observe that $\osem{-} \circ \eta_\Var$
is a coalgebra morphism from $c \colon \Var \rightarrow (A \times \Var)^{A^\Var}$ to the final coalgebra (above diagram;
the fact that $\eta_\Var$ is a coalgebra morphism is again standard).
By finality, it suffices to prove that $\proj$ is a coalgebra morphism. This easily
follows from the definition of $\proj$ and $c$.
\end{proof}
Observe that, by virtue of Theorem \ref{th:pointwise_ext}, the algebra  $\bar{\sigma}$ is the pointwise extension of $\sigma$. This fact will be useful in the next section to relate $\obisim$ with bisimilarity on the original stream system.

\subsection{Abstract, Open and Closed Semantics}
\label{sec:int-openterms}

Recall from Section \ref{sec:preliminaries} the abstract semantics $\asem{-} \colon T A^\omega \rightarrow A^\omega$ arising as in (B) from a monadic abstract stream specification $\lambda$.
The following proposition is the key  to prove Theorem \ref{th:prooftechnique} relating open bisimilarity and abstract semantics.

\begin{proposition}
\label{prop:osem_and_csem}
	Let $\asem{-}$ and $\osem{-}$ be the abstract and open semantics respectively
	of a monadic abstract stream GSOS specification $\lambda$. For any $t \in \openTerms$, $\psi \colon \Var \to A^\omega$:
	$$\osem{t}(\psi) = \asem{(T\psi)(t)} \, .$$
\end{proposition}
\begin{proof}
Let $\psi \colon \Var \to A^\omega$. We formulate what we need to prove
as commutativity of the following diagram:
$$
\begin{tabular}{c}
 \begin{tikzpicture}
  \matrix (m) [matrix of math nodes,row sep=2.5em,column sep=4em]
  {
     T\Var& T A^{\omega}  \\
     \tgamma & A^{\omega} \\
   };     
   \path[-stealth]
     (m-1-1) edge node [above] {\scriptsize $T\psi$} (m-1-2)
             edge node [left] {\scriptsize $\osem{-}$} (m-2-1)
     (m-1-2) edge node [right] {\scriptsize $\asem{-}$} (m-2-2)
     (m-2-1) edge node [above] {\scriptsize $\ev(\psi,-)$} (m-2-2) ;
\end{tikzpicture}
\end{tabular}
$$
where $\ev(\psi,-)$ is defined, for $f \in \tgamma$, by 
\begin{equation}\label{eq:def-ev-psi}
	\ev(\psi,-)(f) = f(\psi) \,.
\end{equation}
It is convenient
below to observe the following equality:
\begin{equation}\label{eq:evpsi}
 \ev(\psi, -) = 
 (\tgamma \simeq 1 \times \tgamma 
 \trans[\psi' \times \id] (A^\Var)^\omega \times \tgamma  
 \trans[\ev] A^\omega)
\end{equation}
 where $\psi' : 1 \to (A^\Var)^\omega$ is a transpose of $\psi$.

Let $\sigma \colon \Sigma A^\omega \rightarrow A^\omega$ be the algebra
 induced by the distributive law $\lambda$.
The proof proceeds by showing that 
$\asem{-} \circ T\psi$ and $\ev(\psi,-) \circ \osem{-}$ are both
 $\Sigma$-algebra morphism from the algebra $\kappa_\Var \colon \Sigma T\Var \rightarrow T\Var$ to $\sigma$,
 such that $\asem{-} \circ T\psi \circ \eta_\Var = \ev(\psi,-)\circ \osem{-} \circ \eta_\Var$. 
 The equality then follows by uniqueness, see Section~\ref{sec:spec}. 
 
 For $\asem{-} \circ T\psi$, consider the following diagram:
\begin{equation}
\begin{tabular}{c}
 \begin{tikzpicture}
  \node at (1.3cm, 0.8cm) {};
  \matrix (m) [matrix of math nodes,row sep=2.5em,column sep=4em]
  {
    \Sigma T\Var&  \Sigma{TA^{\omega}}  & \Sigma A^\omega \\
     T\Var & TA^\omega & A^\omega \\
     V & A^\omega & \\
   };     
    \path[-stealth]
      (m-1-1) edge node [above] {\scriptsize $\Sigma{T\psi}$} (m-1-2)
              edge node [left] {\scriptsize $\kappa_\Var$} (m-2-1)
      (m-2-1) edge node [above] {\scriptsize $T\psi$} (m-2-2)
      (m-3-1) edge node [above]  {\scriptsize $\psi$} (m-3-2)      
              edge node [left] {\scriptsize $\eta_{\Var}$} (m-2-1)
      (m-1-2) edge node [above] {\scriptsize $\Sigma \asem{-}$} (m-1-3)
              edge node [left] {\scriptsize $\kappa_{A^\omega}$} (m-2-2)
      (m-2-2) edge node [above] {\scriptsize $\asem{-}$} (m-2-3)
      (m-3-2) edge node [left]  {\scriptsize $\eta_{A^\omega}$} (m-2-2)
              edge node [below] {\scriptsize $\id$} (m-2-3)      
      (m-1-3) edge node [right] {\scriptsize $\sigma$} (m-2-3) 
     ;
\end{tikzpicture}
\end{tabular}
\end{equation} 
The squares on the left commute by naturality. The triangle on the right below
holds since $\asem{-}$ is an algebra for the (free) monad $T$, and
the upper square is a standard fact in the theory of bialgebras (it follows
from the definition of $\sigma$ and that $\asem{-}$ is an algebra for the monad $T$).
This  shows that $\asem{-} \circ T\psi$ is an algebra morphism, and
$\asem{-} \circ T\psi \circ \eta_\Var = \psi$.
 
For $\ev(\psi,-) \circ \osem{-}$, consider the diagram below.
Here $\st^{\Sigma}_{(A^\Var)^\omega,\tgamma} \colon (A^\Var)^\omega \times \Sigma\tgamma \rightarrow \Sigma((A^\Var)^\omega \times \tgamma)$
is the \emph{strength} of the functor $\Sigma$, see \ref{sec:strengthco} for some basic properties which are used in the rest of the proof.
\begin{equation}
\begin{tabular}{c}
 \begin{tikzpicture} 
  \node at (-4.9cm, 0.7cm) {\scriptsize\textit{(i)}};
  \node at (-5cm, -1.3cm) {\scriptsize\textit{(ii)}};  
  \node at (-3cm, 0.7cm) {\scriptsize\textit{nat. iso.}};  
  \node at (0.5cm, 1.6cm) {\scriptsize\textit{nat. }$\st^\Sigma_{-,\tgamma}$};    
  \node at (3.5cm, 0.7cm) {\scriptsize\textit{pointw. ext.}};    
  \node at (-1.9cm, 2.8cm) {\scriptsize (\ref{eq:strength})};
  \node at (0.8cm, -1.6cm) {\scriptsize (\ref{eq:evpsi})};
  \matrix (m) [matrix of math nodes,row sep=2.5em,column sep=2em]
  {
    \Sigma T\Var&  \Sigma{\tgamma}  &  1 \times \Sigma{\tgamma} & 
    \Sigma(1 \times {\tgamma}) &  \Sigma((A^\Var)^\omega \times {\tgamma})
    & \Sigma A^\omega \\
    & & & (A^\Var)^\omega \times \Sigma {\tgamma} & & \\
    T\Var&  \tgamma  &  1 \times \tgamma & 
    (A^\Var)^\omega \times \tgamma  & & A^\omega \\
    \Var & & & & & \\
   };     
     \path[-stealth]
       (m-1-1) edge node [above] {\scriptsize } (m-1-2)
               edge node [left] {\scriptsize $\kappa_V$} (m-3-1)
       (m-3-1) edge node [above] {\scriptsize $\osem{-}$} (m-3-2)
       (m-4-1) edge node [left]  {\scriptsize $\eta_\Var$} (m-3-1)      
               edge node [below right]  {\scriptsize $\proj$} (m-3-2)
       (m-1-2) edge node [above] {\scriptsize $\simeq$} (m-1-3)
               edge node [right] {\scriptsize $\bar \sigma$} (m-3-2)
               edge [bend left=45] node [above]{\scriptsize $\Sigma \simeq$} (m-1-4)
       (m-3-2) edge node [above] {\scriptsize $\simeq$} (m-3-3)
               edge [bend right=25] node [below]{\scriptsize $\ev(\psi,-)$} (m-3-6) 
       (m-1-3) edge node [above]  {\scriptsize $\st^{\Sigma}_{1,\tgamma}$} (m-1-4)
               edge node [left=1mm] {\scriptsize $\psi' \times \id$} (m-2-4)
               edge node [right] {\scriptsize $\id \times \bar \sigma$} (m-3-3)
       (m-3-3) edge node [above]  {\scriptsize $\psi' \times \id$} (m-3-4)
       (m-1-4) edge node [above]  {\scriptsize $\Sigma(\psi' \times \id)$} (m-1-5)
       (m-2-4) edge node [right=1mm] {\scriptsize $\st^{\Sigma}_{(A^\Var)^\omega,\tgamma}$} (m-1-5)
               edge node [right]  {\scriptsize $\id \times \bar \sigma$} (m-3-4)
       (m-3-4) edge node [above]  {\scriptsize $\ev$} (m-3-6)        
       (m-1-5) edge node [above]  {\scriptsize $\Sigma\ev$} (m-1-6)
       (m-1-6) edge node [right]  {\scriptsize $\sigma$} (m-3-6)
      ;
\end{tikzpicture}
\end{tabular}
\end{equation} 
where $(i),(ii)$ commute by Lemma~\ref{lm:semantics} (see also the lemma for the definition of $\proj$). Hence $\ev(\psi,-) \circ \osem{-}$
is an algebra morphism, and it only remains to prove that
$\ev(\psi,-) \circ \proj = \psi$:
\begin{align*}
\ev(\psi,-) \circ \proj(\vx)
& = \proj(\vx)(\psi) & \eqref{eq:def-ev-psi} \\ 
& = \psi(\vx) & \text{def. of }\proj 
\end{align*}
From $\ev(\psi,-) \circ \osem{-} \circ \eta_\Var = \ev(\psi,-) \circ \proj = \psi = \csem{-} \circ T\psi \circ \eta_\Var$
and the fact (shown above) that $\ev(\psi,-) \circ \osem{-}$ and $\csem{-} \circ T\psi$ are algebra morphisms of
the same type, we conclude that they are equal.
\end{proof}

As a simple consequence, we obtain the following characterization of $\obisim$.
\begin{theorem}
  \label{th:prooftechnique}
	For all $t_1, t_2 \in T\Var$, 
 $\osem{t_1} = \osem{t_2}$  iff for all $\psi\colon  \Var \to A^\omega$: $\asem{T\psi(t_1)} = \asem{T\psi(t_2)}$.
\end{theorem}

 This is one of the main results of this paper: $T\psi(t_1)$ and $T\psi(t_2)$ are expressions in $TA^{\omega}$ built from symbols of the signature $\Sigma$ and streams $\alpha_1, \dots \alpha_n \in A^\omega$. By checking $t_1 \obisim t_2$ one can prove that the two expressions are equivalent for all possible streams $\alpha_1, \dots \alpha_n \in A^\omega$.

\begin{example}
By using the Mealy machine $m_\lambda$ in Example~\ref{ex:alt-ex-open}, the relation
\begin{align*}
 R = &\{(\alt(\vx, \alt(\vy,\vz)),\alt(\vx, \alt(\vw,\vz))), (\alt(\alt(\vz,\vy),\vx),\alt(\alt(\vz,\vw),\vx))\}
\end{align*}
is easily verified to be a bisimulation (Definition \ref{def:bs:mmachine}). In particular this shows that $\osem{(\alt(\vx, \alt(\vy,\vz))} = \osem{ \alt(\vx, \alt(\vw,\vz))}$. By Theorem~\ref{th:prooftechnique}, we have that $\asem{T\psi(\alt(\vx, \alt(\vy,\vz))} = \asem{T\psi(\alt(\vx, \alt(\vw,\vz)))}$ for all $\psi\colon  \Var \to A^\omega$, i.e.,
$$ \alt(\alpha_1, \alt(\alpha_2,\alpha_3)) \bisim \alt(\alpha_1, \alt(\alpha_4,\alpha_3)) \text{ for all } \alpha_1,\alpha_2,\alpha_3, \alpha_4 \in A^\omega\text{.}$$
The above law can be understood as an equivalence of program schemes stating that one can always replace the stream $\alpha_2$ by an arbitrary stream $\alpha_4$, without changing the result.
\end{example}

\begin{example}\label{ex:acommutativity}
By using the Mealy machines in Example \ref{ex:streamcalculus-open}, it is easy to check that both $\{ ((\vx \oplus \vy)  \oplus \vz , \vx \oplus ( \vy  \oplus \vz )) \}$ and $\{ ( \vx \oplus \vy  ,  \vy \oplus \vx) \}$ are bisimulations. This means that $\osem{(\vx \oplus \vy)  \oplus \vz} = \osem{ \vx \oplus ( \vy  \oplus \vz )}$ and $\osem{\vx \oplus \vy}= \osem{  \vy \oplus \vx}$. By Theorem \ref{th:prooftechnique} we obtain associativity and commutativity of $\oplus$:
$$(\alpha_1 \oplus \alpha_2)  \oplus \alpha_3 \sim \alpha_1 \oplus ( \alpha_2  \oplus \alpha_3 ) \text{ and }   \alpha_1 \oplus \alpha_2  \sim  \alpha_2 \oplus \alpha_1  \text{ for all } \alpha_1,\alpha_2,\alpha_3 \in A^\omega \text{.}$$
\end{example}

\begin{example}\label{ex:prefixsum}
In a similar way, one can check that  $\{((a+b).(\vx \oplus \vy), a.\vx \oplus b.\vy) \mid a,b \in \R \}$ is a  bisimulation. This means that $\osem{(a+b).(\vx \oplus \vy)} = \osem{ a.\vx \oplus b.\vy}$ for all $a,b \in \R $ and, using again Theorem \ref{th:prooftechnique}, we conclude that 
$(a+b).(\alpha_1 \oplus \alpha_2)  \sim a.\alpha_1 \oplus b.\alpha_2  \text{ for all } \alpha_1,\alpha_2 \in A^\omega  \text{.}$
\end{example}

Often, equivalence of open terms is defined by relying on the equivalence of closed terms: two open terms are equivalent iff under all possible closed substitutions, the resulting closed terms are equivalent. For $\obisim$, this property does not follow immediately  by Theorem \ref{th:prooftechnique}, where variables range over streams, i.e., elements of the final coalgebra.  One could assume that all the behaviours of the final coalgebra are denoted by some term, however this restriction would rule out most of the languages we are aware of: in particular, the stream calculus that can express only the so-called rational streams \cite{Rutten05}.

The following theorem, which is the second main result of this paper, only requires that the stream GSOS specification is sufficiently expressive to describe arbitrary finite prefixes. We use that any closed substitution $\phi \colon \Var \to \closedTerms$ defines  $\phi^\dagger \colon \openTerms \to \closedTerms$ (see Section~\ref{sec:spec}).
\begin{theorem}\label{th:closed-terms}
	Suppose $\lambda \colon  \Sigma (A \times {-}) \Rightarrow A \times \tsigma$ is a monadic abstract stream GSOS specification
	which contains, for each $a \in A$, the prefix operator $a.-$ as specified in~\eqref{rule:prefix} in
	Section~\ref{sec:from_full_to_monadic}. Further, assume $\closedTerms$ is non-empty. 
	
	Let $\csem{-}$ and $\osem{-}$ be the closed and open semantics respectively of $\lambda$. 
	Then for all $t_1, t_2 \in T\Var$: 
 $\osem{t_1} = \osem{t_2}$  iff $\csem{\phi^\dagger(t_1)} = \csem{\phi^\dagger(t_2)}$ for all $\phi\colon  \Var \to \closedTerms$.
\end{theorem}

\begin{proof}
	First, we prove that for any $\phi \colon \Var \to \closedTerms$, the following diagram commutes:
	\begin{equation}\label{eq:csem-asem-sub}
	\xymatrix{
		T\Var \ar[r]^{T \phi} \ar[dr]_{\phi^\dagger}
			& TT\emptyset \ar[r]^{T \csem{-}} \ar[d]^{\mu_\emptyset}
			& T A^\omega \ar[d]^{\asem{-}} \\
		& 	T\emptyset \ar[r]_{\csem{-}}
		& A^\omega
	}
	\end{equation}
	Indeed, commutativity of the left-hand side is a general property of monads, commutativity of the right-hand side a 
	general property of algebras induced by distributive laws. 
	
	Now, for the implication from left to right, 
	suppose $\osem{t_1} = \osem{t_2}$, and let $\phi \colon \Var \to \closedTerms$. Then
	$$
		\csem{\phi^\dagger(t_1)} = \asem{T(\csem{-} \circ \phi)(t_1)} = \asem{T(\csem{-}\circ \phi)(t_2)} = \csem{\phi^\dagger(t_2)} \,.
	$$
	The first and last equality hold by~\eqref{eq:csem-asem-sub}. The middle equality holds
	by Theorem~\ref{th:prooftechnique}, instantiated to $\psi = \csem{-} \circ \phi \colon \Var \rightarrow A^\omega$.
	
	For the converse, suppose 
	\begin{equation}\label{eq:assump-fdagger}
		\csem{\phi^\dagger(t_1)} = \csem{\phi^\dagger(t_2)}
	\end{equation} for all $\phi: \Var \to \closedTerms$,
	and let $\psi \colon \Var \rightarrow A^\omega$. We need to prove that $\osem{t_1}(\psi) = \osem{t_2}(\psi)$.
	Let $n \geq 0$, and define $\phi_n \colon \Var \rightarrow T\emptyset$ by
	$$
	\phi_n(\vx) = \psi(\vx)(0).\psi(\vx)(1). (\ldots) .\psi(\vx)(n-1).t
	$$
	for some $t \in T\emptyset$ (assumed to exist); this is indeed a term in $T\emptyset$, since
	we assumed the presence of the prefix operator.
	By definition of the prefix operator, it follows that
	\begin{equation}\label{eq:proof-approx}
		\csem{\phi_n(\vx)} {\upharpoonright_n} = \psi(\vx) {\upharpoonright_n}
	\end{equation}
	i.e., the first $n$ elements of $\csem{\phi_n(\vx)}$ coincide with the first $n$ elements of $\psi(\vx)$.
	For all $t \in \openTerms$, since $\osem{t}$ is causal, we have
	\begin{equation}\label{eq:t-eq-upton}
		\osem{t}(\csem{-} \circ \phi_n) {\upharpoonright_n} = \osem{t}(\psi){\upharpoonright_n} \,.
	\end{equation}
	Further, we have, for any $t \in \openTerms$,
	\begin{equation}\label{eq:some-osem-csem-stuff}
		\osem{t}(\csem{-} \circ \phi_n) 
			 = \asem{T(\csem{-} \circ \phi_n)(t)} 
			 = \csem{\phi_n^\dagger(t)}
	\end{equation}
	by Proposition~\ref{prop:osem_and_csem} and~\eqref{eq:csem-asem-sub} respectively.
	We obtain 
	\begin{align*}
	  \osem{t_1}(\psi){\upharpoonright_n} 
	  = \osem{t_1}(\csem{-} \circ \phi_n) {\upharpoonright_n}
	  &= \csem{\phi_n^\dagger(t_1)}{\upharpoonright_n} \\
	  &= \csem{\phi_n^\dagger(t_2)}{\upharpoonright_n}
	  = \osem{t_2}(\csem{-} \circ \phi_n) {\upharpoonright_n}
	  =\osem{t_2}(\psi){\upharpoonright_n} 
	\end{align*}
	by~\eqref{eq:t-eq-upton},~\eqref{eq:some-osem-csem-stuff} and assumption~\eqref{eq:assump-fdagger}.
	Since this works for all $n$, we conclude $\osem{t_1}(\psi) = \osem{t_2}(\psi)$ as desired.
\end{proof}

\begin{example}
The specification in Figure \ref{fig:rules} does not include the prefix operator, therefore it does not meet the assumptions of Theorem \ref{th:closed-terms}.
Instead, the monadic GSOS specification in Figure \ref{fig:rulestreamcalculus}(b) contains the prefix. Recall from Example \ref{ex:prefixsum} that $(a+b).(\vx \oplus \vy) \obisim a.\vx \oplus b.\vy$. Using Theorem \ref{th:closed-terms}, we can conclude that
$(a+b).(t_1 \oplus t_2) \bisim a.t_1 \oplus b.t_2 \text{ for all } t_1,t_2\in \closedTerms \text{.}$
\end{example}


\section{Bisimulation up-to substitutions}
\label{sec:compatibility}


In the previous section, we have shown that bisimulations on Mealy machines can be used to prove equivalences of open terms specified in the stream GSOS format. In this section we introduce \emph{up-to substitutions}, an enhancement of  the bisimulation proof method that allows to deal with smaller, often finite, relations. We also show that up-to substitutions can be effectively combined with other well-known up-to techniques such as up-to bisimilarity and up-to context.
We note that our results here strongly rely on the specification being monadic:
in \ref{app:is_not_compatible} we show that they fail in general for non-monadic specifications.

Intuitively, in a bisimulation up-to substitutions ${\relR}$, the states reached by a pair of states do not need to be related by ${\relR}$, but rather by 
${\theta(\relR})$, for some substitution $\theta \colon \Var \to \openTerms$.
We give a concrete example. Suppose we extend the stream calculus of Example~\ref{ex:stream_calculus} with
the operators $f$ and $g$ defined by the rules in Figure~\ref{fig:f_and_g}.
In Figure~\ref{fig:f_and_g:ext}, we have the pointwise extensions of these new operators. 
It should be clear that $f(\vx) \bisim g(\vx)$. 
To try to formally prove $ f(\vx) \bisim g(\vx)$, consider the relation 
${\relR} = \{(f(\vx), g(\vx))\}$.
For all $\subs \colon \Var \to A$, there are transitions
$f(\vx) \trans[\subs | \subs(\vx)] f(\vx \oplus \vx)$ and $g(\vx) \trans[\subs | \subs(\vx)] g(\vx \oplus \vx)$. 
The outputs of both transitions coincide but the reached states are not in ${\relR}$, 
hence ${\relR}$ is not a bisimulation. However it is a bisimulation up-to substitutions, since the arriving states are related by ${\theta(\relR})$, for some substitution $\theta$ mapping $\vx$ to $\vx \oplus  \vx$. In fact, without this technique, any bisimulation relating $f(\vx)$ and $g(\vx)$ should contain infinitely many pairs.

\begin{figure}[t]
  \centering
  \makebox[0pt][c]{\parbox{1\textwidth}{%
      \begin{minipage}[b]{0.49\hsize}\centering
      \begin{gather*}
       \dedrule{x \trans[a] x'}{f(x) \trans[a] f(x'\oplus x')}
       \     
       \dedrule{x \trans[a] x'}{g(x) \trans[a] g(x'\oplus x')}
      \end{gather*}
      
      \caption{$f$ and $g$, operators over streams}
      \label{fig:f_and_g}
    \end{minipage}
    \begin{minipage}[b]{0.49\hsize}\centering
    \begin{gather*}
    \dedrule{x \trans[\subs | a] x'}{f(x) \trans[\subs |a] f(x' \oplus x')}
    \ 
    \dedrule{x \trans[\subs |a] x'}{g(x) \trans[\subs |a] g(x'\oplus x')}
    \end{gather*}

    \caption{Pointwise extensions of $f$ and $g$.}
    \label{fig:f_and_g:ext}
    \end{minipage}
}}
\end{figure}

In order to prove the soundness of this technique, as well as the fact that it can be safely combined with other known up-to techniques, we need to recall some notions of the theory of up-to techniques in lattices from~\cite{PS11}.
Given a Mealy machine $(X,m)$, we consider the lattice $(\mathcal{P}(X\times X), \subseteq)$ of relations over $X$, ordered by inclusion, and the monotone map $\b \colon \mathcal{P}(X\times X) \to \mathcal{P}(X\times X)$ defined for all $ {\relR} \subseteq X\times X$ as
\begin{equation}
\label{eq:def:b}
\b(\relR) = \{(s,t) \in X\times X \mid  \forall b \in B, o_s(b)=o_t(b) \text{ and } d_s(b) \relR d_t(b)\}\text{.}
\end{equation}
It is easy to see that post fixed points of $\b$, i.e., relations $\relR$ such that ${\relR} \subseteq \b(\relR)$, are exactly bisimulations for Mealy machines (Definition~\ref{def:bs:mmachine}) and that its greatest fixed point is $\sim$. 

For a monotone map $\f \colon \mathcal{P}(X\times X) \to \mathcal{P}(X\times X)$, a \emph{bisimulation up-to $\f$} is a relation ${\relR}$ such that ${\relR} \subseteq \b \f(\relR)$. We say that $\f$ is \emph{$\b$-compatible} if ${\f \b(\relR)} \subseteq \b \f (\relR)$ for all relations ${\relR}$. Two results in~\cite{PS11} are pivotal for us: first, if $\f$ is compatible and ${\relR} \subseteq \b\f(\relR)$ then ${\relR} \subseteq {\sim}$; second if $\f_1$ and $\f_2$ are $\b$-compatible then $\f_1\circ \f_2$ is $\b$-compatible. The first result informs us that bisimilarity can be proved by means of bisimulations up-to $\f$, whenever $\f$ is compatible. The second result states that compatible up-to techniques can be composed. 

We now consider up-to techniques for the Mealy machine over open terms $(\openTerms,m_\lambda)$ as defined in Section \ref{sec:mm_over_ot}. Recall that bisimilarity over this machine is called open bisimilarity, denoted by $\obisim$.
Up-to substitutions is the monotone function $(-)_{\forall\theta} \colon \mathcal{P}( \openTerms\times \openTerms) \to \mathcal{P}( \openTerms\times \openTerms)$ mapping a relation ${\relR}\subseteq \openTerms\times \openTerms$ to
\begin{equation*}
(\relR)_{\forall\theta} = \{(\theta(t_1),\theta(t_2)) \mid \theta \colon \Var \to \openTerms \text{ and } t_1 \relR t_2 \}\text{.}
\end{equation*}
Similarly, we define up-to context as the monotone function mapping every relation ${\relR} \subseteq \openTerms\times \openTerms$ to its contextual closure $\mathcal C(\relR)$ and up-to (open) bisimilarity as the function mapping $\relR$ to ${\obisim} \relR {\obisim} = \{(t_1,t_2) \mid \exists t_1',t_2' \text{ s.t. } t_1 \obisim t_1' \relR t_2' \obisim t_2 \}$.

Compatibility with $\b$ of up-to context and up-to bisimilarity hold immediately by the results in \cite{BonchiPPR17}. For up-to substitutions, we will next prove compatibility (Theorem~\ref{th:upto}). First,
we prove a technical lemma, which states a relation between the derivative of a term after a substitution for a particular input 
and the derivative of the original term w.r.t.\ an input constructed based on the substitution and the first input. 
We use the following notation: given $t \in \openTerms$, $t(\vx_1, \ldots, \vx_n)$ 
denotes that $\Var(t) \subseteq \{\vx_1, \ldots, \vx_n\}$ and $t(t_1, \ldots, t_n)$ the term obtained from
$t$ by replacing $\vx_i$ by $t_i$.

\begin{lemma}
\label{lemma:aux}
 Let $m : \openTerms \to (A \times \openTerms)^{A^{\Var}}$ be a Mealy machine induced by
 a monadic Mealy GSOS specification that is the pointwise extension of a monadic stream GSOS specification.
 Consider a substitution $\theta \colon \Var \to \openTerms$ and an input $\subs \colon \Var \to A$.
 Then there exists a subsitution $\theta' \colon \Var \to \openTerms$ and input $\subs \colon \Var \to A$ such that $o_{\theta(t)}(\subs) = o_t(\subs')$ and $d_{\theta(t)}(\subs) = \theta'(d_{t}(\subs'))$ for all terms $t \in \openTerms$.
 \footnote{Recall $m(t) = \tuple{o_t, d_t} : A^{\Var} \to A \times \openTerms$ where $o_t(\subs)$ and $d_t(\subs)$ are respectively the output and the derivative of $t \in \openTerms$ for the input $\subs \in A^{\Var}$.}
\end{lemma}
\begin{proof}
 We prove this by induction on the structure of the term $t$,
 with $\theta' \colon \Var \to \openTerms$ and $\subs' \colon \Var \to A$ given by
 \begin{enumerate}[(i)]
   \item $\theta'(\vx) = d_{\theta(\vx)}(\subs)$ for all $\vx \in \Var$, and
   \item $\subs'(\vx) = o_{\theta(\vx)}(\subs)$ for all $\vx \in \Var$.
 \end{enumerate}
 First we show that the base case holds, i.e. $t = \vx$ for some $\vx \in \Var$. 
 First observe that $o_{\vx}(\subs) = \subs(\vx)$ and $d_{\vx}(\subs) = \vx$
 for every input $\subs \colon \Var \rightarrow A$,
 as on variables the Mealy machine $m$ is fully defined by the Mealy machine $c$.
 From this it follows that $o_{\theta(\vx)}(\subs) = \subs'(\vx) = o_{\vx}(\subs')$,
 and that $d_{\theta(\vx)}(\subs) = \theta'(\vx) = \theta'(d_{\vx}(\subs'))$.
 
 Next we prove the inductive case.
 Consider the case $t = f(t_1, \ldots, t_n)$ where each $t_i$ satisfies the inductive hypothesis.
 For all $a_1,\ldots,a_n \in A$ and $\subs \in A^\Var$ there is a rule in the Mealy GSOS specification with the following shape:
 \begin{equation}
 \label{eq:rule_in_the_proof}
 \dedrule{\vx_1 \trans[\subs|a_1] \vx'_1 \quad \ldots \quad \vx_n \trans[\subs|a_n] \vx'_n}
 {f(\vx_1, \ldots,\vx_n) \trans[\subs|a^f_{a_1, \ldots, a_n}] t^f_{a_1, \ldots, a_n} (\vx'_1, \ldots,\vx'_n)} 
 \end{equation}
 i.e. the output in the conclusion for $f$ only depends of the outputs of the premises.
 First we show that $o_{\theta(f(t_1,\ldots,t_n))}(\subs) = o_{f(t_1,\ldots,t_n)}(\subs')$ using the following calculation:
 \begin{align*}
  & o_{\theta(f(t_1,\ldots,t_n))}(\subs) \\
  & = \{\text{by the def. of substitution on terms}\}\\
  & o_{f(\theta(t_1),\ldots,\theta(t_n))}(\subs) \\
  & = \{\text{by the appropriate instantiation of rule (\ref{eq:rule_in_the_proof})}\}\\
  & a^f_{o_{\theta(t_1)}(\subs),\ldots,o_{\theta(t_n)}(\subs)} \\
  & = \{\text{by the induction hypothesis}\} \\
  & a^f_{o_{t_1}(\subs'),\ldots,o_{t_n}(\subs')} \\
  & = \{\text{by the appropriate instantiation of rule (\ref{eq:rule_in_the_proof})}\}\\
  & o_{f(t_1,\ldots,t_n)}(\subs')
 \end{align*}
 Finally, $d_{\theta(f(t_1, \ldots, t_n))}(\subs) = \theta'(d_{f(t_1, \ldots, t_n)}(\subs'))$ follows from another calculation:
 \begin{align*}
  & d_{\theta(f(t_1,\ldots, t_n))}(\subs) \\
  & = \{\text{by def. of substitution on terms}\} \\
  & d_{f(\theta(t_1),\ldots, \theta(t_n))}(\subs) \\
  & = \{\text{by the appropriate instantiation of rule (\ref{eq:rule_in_the_proof})}\} \\  
  & t^f_{o_{\theta(t_1)}(\subs), \ldots, o_{\theta(t_n)}(\subs)}(d_{\theta(t_1)}(\subs), \ldots, d_{\theta(t_n)}(\subs)) \\  
  &  = \{\text{by induction, $o_{\theta(t_i)}(\subs) = o_{t_i}(\subs')$ for $i = 1, \ldots, n$ }\} \\ 
  & t^f_{o_{t_1}(\subs'), \ldots, o_{t_n}(\subs')}(d_{\theta(t_1)}(\subs), \ldots, d_{\theta(t_n)}(\subs)) \\
  &  = \{\text{by induction, $d_{\theta(t_i)}(\subs) = \theta'(d_{t_i}(\subs'))$ for $i = 1, \ldots, n$ }\} \\ 
  & t^f_{o_{t_1}(\subs'), \ldots, o_{t_n}(\subs')}(\theta'(d_{t_1}(\subs')), \ldots, \theta'(d_{t_n}(\subs')))\\
  & = \{\text{by def. of substitution on terms}\} \\
  & \theta'(t^f_{o_{t_1}(\subs'), \ldots, o_{t_n}(\subs')}(d_{t_1}(\subs'), \ldots, d_{t_n}(\subs')))\\
  & = \{\text{by the appropriate instantiation of rule (\ref{eq:rule_in_the_proof})}\} \\
  & \theta'(d_{f(t_1, \ldots, t_n)}(\subs'))
 \end{align*} 
\end{proof}

\begin{theorem}
 \label{th:upto}
 The function $(-)_{\forall\theta}$ is $\b$-compatible.
\end{theorem}

\begin{proof}
  For convencience we will write $\f = (-)_{\forall\theta}$.
  Let $R \subseteq \openTerms \times \openTerms$ be a relation,
  and let $(\theta(s), \theta(t)) \in \f(\b(R))$ for some substitution $\theta \colon \Var \to \openTerms$ and pair of terms $(s, t) \in \b(R)$.
  In order to show that $(\theta(s), \theta(t)) \in \b(\f(R))$ we need to prove that
  \begin{enumerate}[(i)]
      \item $o_{\theta(s)}(\subs) = o_{\theta(t)}(\subs)$, and
      \item $(d_{\theta(s)}(\subs), d_{\theta(t)}(\subs)) \in \f(R)$,
  \end{enumerate}
  for every input $\subs \colon \Var \rightarrow A$.

  Now fix $\subs$.
  By Lemma~\ref{lemma:aux} there exists a $\subs' \colon \Var \to A$,
  such that $o_{\theta(t)}(\subs) = o_t(\subs')$ for all $t \in \openTerms$.
  Then, as $(s, t) \in \b(R)$, we have
  \begin{equation*}
    o_{\theta(s)}(\subs) = o_s(\subs') = o_t(\subs') = o_{\theta(t)}(\subs).
  \end{equation*}
  Next we need to show that $(d_{\theta(s)}(\subs), d_{\theta(t)}(\subs)) \in \f(R)$,
  i.e.\ that there exists a substitution $\theta' \colon \Var \to \openTerms$
  and terms $s', t' \in \openTerms$ such that $\theta'(s') = d_{\theta(s)}(\subs)$,
  $\theta'(t') = d_{\theta(t)}(\subs)$, and $(s', t') \in R$.
  By Lemma~\ref{lemma:aux} there exist $\theta' \colon \Var \to \openTerms$ and $\subs' \colon \Var \to A$
  such that $d_{\theta(t)}(\subs) = \theta'(d_t(\subs'))$ for all $t \in \openTerms$,
  and thus
  \begin{equation*}
    (d_{\theta(s)}(\subs), d_{\theta(t)}(\subs)) = (\theta'(d_s(\subs')), \theta'(d_t(\subs'))).
  \end{equation*}
  By the assumption that $(s, t) \in \b(R)$, we get that $(d_s(\subs'), d_t(\subs')) \in R$.
  So we can conclude that $(d_{\theta(s)}(\subs), d_{\theta(t)}(\subs)) \in \f(R)$.
\end{proof}

As a consequence of the above theorem and the results in \cite{PS11}, up-to substitutions can be used in combination with up-to bisimilarity and up-to context (as well as any another compatible up-to technique) to prove open bisimilarity. We will show this in the next, concluding example, for which a last remark is useful: the theory in \cite{PS11} also ensures that if $\f$ is  $\b$-compatible, then $\f(\sim) \subseteq {\sim}$. By Theorem \ref{th:upto}, this means that $(\obisim)_{\forall\theta} \subseteq {\obisim}$. The same obviously holds for the contextual closure: $\mathcal C(\obisim)\subseteq {\obisim}$.

\begin{example}\label{ex:uptosubs}
We prove that the convolution product $\otimes$ distributes over the sum $\oplus$, i.e.,
$\alpha_1 \otimes (\alpha_2 \oplus \alpha_3) \sim (\alpha_1 \otimes \alpha_2) \oplus (\alpha_1 \otimes \alpha_3) \text{ for all streams } \alpha_1,\alpha_2,\alpha_3 \in \R^\omega \text{.}$
By Theorems \ref{th:prooftechnique} and \ref{th:upto}, to prove our statement it is enough to show that  
${\relR} = \{(\vx \otimes (\vy \oplus \vz) , (\vx \otimes \vy) \oplus (\vx \otimes \vz))\}$
is a bisimulation up-to  ${\obisim} \mathrel{\mathcal{C}(\obisim (-)_{\forall \theta}\obisim)}{\obisim}$.

By rules in Figure~\ref{fig:rulestreamcalculus}{(d)}, for all $\subs \colon \Var \to \R$,
the transitions of the open terms are


\begin{itemize}
 \item 
 $\vx \otimes (\vy \oplus \vz) \trans[\subs | \subs(\vx) \times (\subs(\vy) + \subs(\vz))] ( \subs(\vx) \otimes (\vy \oplus \vz) ) \oplus ( \vx \otimes (\subs(\vy) + \subs(\vz)).(\vy \oplus \vz) )$
 \item 
 $(\vx \otimes \vy) \oplus (\vx \otimes \vz) \trans[\subs | \subs(\vx) \times \subs(\vy) + \subs(\vx) \times \subs(\vz)]$ \\ 
 $((\subs(\vx) \otimes \vy) \oplus (\vx \otimes \subs(\vy).\vy)) \oplus ( (\subs(\vx) \otimes \vz) \oplus (\vx \otimes \subs(\vz).\vz) )$
\end{itemize}

For the outputs, it is evident that $\subs(\vx) \times (\subs(\vy) + \subs(\vz)) =  \subs(\vx) \times \subs(\vy) + \subs(\vx) \times \subs(\vz)$.
For the arriving states we need a few steps, where for all $\subs \colon \Var \to \R$ and $\vx\in \Var$, $\subs (\vx)$ denotes either a real number (used as a prefix) or a constant of the syntax (Example \ref{ex:stream_calculus}).

 \begin{enumerate}[(a)]

  \item
  \label{st3}
  $\vx \otimes (\subs(\vy).\vy \oplus \subs(\vz).\vz) 
  \mathrel{{\relR}_{\forall \theta}} (\vx \otimes \subs(\vy).\vy ) \oplus (\vx \otimes \subs(\vz).\vz)$.

  \item 
  \label{st2}
  By Example \ref{ex:prefixsum} and $\mathcal C(\obisim)\subseteq \obisim$, we have that:\\
  $\vx \otimes (\subs(\vy) + \subs(\vz)).(\vy \oplus \vz) \obisim \vx \otimes (\subs(\vy).\vy \oplus \subs(\vz)).\vz)$.
  
      \item 
  \label{st4}
  By (\ref{st2}) and (\ref{st3}):\\
  $\vx \otimes (\subs(\vy) + \subs(\vz)).(\vy \oplus \vz) 
  \mathrel{{\obisim}{\relR}_{\forall \theta}{\obisim}}\ \!
  (\vx \otimes \subs(\vy).\vy ) \oplus (\vx \otimes \subs(\vz).\vz)$.

    \item 
  \label{st1}
  $\subs(\vx) \otimes (\vy \oplus \vz) \mathrel{\relR_{\forall \theta}} (\subs(\vx) \otimes \vy) \oplus (\subs(\vx) \otimes \vz)$.

  \item 
  \label{st5}
  Using (\ref{st1}) and (\ref{st4}) with context $\mathcal{C} = \_ \oplus \_$: \\
  $(\subs(\vx) \otimes (\vy \oplus \vz)) \oplus (\vx \otimes (\subs(\vy) + \subs(\vz)).(\vy \oplus \vz)) \\
  \mathrel{\mathcal{C}(\obisim\relR_{\forall \theta}\obisim)}~
  ((\subs(\vx) \otimes \vy) \oplus (\subs(\vx) \oplus \vz)) \oplus ((\vx \otimes \subs(\vy).\vy ) \oplus (\vx \otimes \subs(\vz).\vz))$.
  \item 
  \label{st6}
  By Example \ref{ex:acommutativity} (associativity and commutativity of $\oplus$) and $(\obisim)_{\forall\rho} \subseteq {\obisim}$:\\
  $((\subs(\vx) \otimes \vy) \oplus (\subs(\vx) \oplus \vz)) \oplus ((\vx \otimes \subs(\vy).\vy ) \oplus (\vx \otimes \subs(\vz)).\vz)) \\
  \obisim 
  ((\subs(\vx) \otimes \vy) \oplus (\vx \otimes \subs(\vy).\vy)) \oplus ((\subs(\vx) \otimes \vz) \oplus (\vx \otimes \subs(\vz).\vz))
  $.
 \item \label{st7}
 By (\ref{st5}) and (\ref{st6}):\\
 $(\subs(\vx) \otimes (\vy \oplus \vz) ) \oplus ( \vx \times (\subs(\vy) + \subs(\vz)).(\vy \oplus \vz) ) \\
 \mathrel{{\obisim}\mathcal{C}(\obisim\relR_{\forall \theta}\obisim){\obisim}}
 ((\subs(\vx) \otimes \vy) \oplus (\vx \otimes \subs(\vy).\vy) ) \oplus ( (\subs(\vx) \otimes \vz) \oplus (\vx \otimes \subs(\vz).\vz) )$.
 \end{enumerate}
\end{example}

\section{A ``familiar'' construction of the Mealy Machine of Open Terms}\label{sec:familiar}

In the previous sections, we transformed a monadic abstract GSOS specification into a `Mealy machine
of open terms', using the pointwise extension. 
While it is still unclear if---and how---our approach generalises
beyond streams and Mealy machines to arbitrary coalgebras, in the current section we make the
first steps toward such a generalisation by connecting our
work to recent developments in the theory of distributive laws~\cite{PousR17,BasoldPR17}.
These developments provide a `semantics of distributive laws',
by organising them into a category with a final object, the so-called \emph{companion}. 

The main aim of this section is to connect the concrete construction
that we provided in the current paper to this abstract work on distributive laws 
and the companion. On the one hand, we show that the construction of a
`Mealy machine of open terms' and the association of causal functions to
open terms (Section~\ref{sec:mm_over_ot}), 
is an instance of a more general construction in the theory of distributive laws. 
On the other hand, it provides a concrete case study for some
of the abstract techniques in~\cite{PousR17,BasoldPR17}.
We therefore view the technical development in the current section as a potentially useful
step towards a more general coalgebraic theory of open terms.

\medskip

The connection with~\cite{PousR17,BasoldPR17} can be anticipated in a nutshell.  
Let $[\set,\set]$ be the category of $\set$ endofunctors and natural transformations between them. 
In~\cite{PousR17,BasoldPR17}, a functor $\Fam{F} \colon [\set, \set] \rightarrow [\set, \set]$ is defined
with the property that $\Fam{F}$-coalgebras are in bijective correspondence with distributive laws over a 
fixed functor $F\colon\set \to \set$
(we use the blackboard
font $\Fam{F}$ to distinguish from the $\set$ endofunctors used throughout the paper, denoted by plain capital letters).
When instantiated to the functor for stream systems $FX = A \times X$,  a distributive law $\rho\colon T F \Rightarrow F T$, corresponds via the above construction to an $\Fam{F}$-coalgebra, i.e., a natural transformation  $\widehat\rho \colon T \Rightarrow \Fam{F}(T)$. Interestingly, this natural transformation $\widehat{\rho}$, at the component $\Var$, turns out to be the Mealy machine $m_\lambda$ of open terms defined in \eqref{eq:mm-ot}.
In this way, the Mealy machine $m_\lambda$ is thus constructed via the correspondence between distributive laws
and coalgebras.  

In the remainder of this section we make this more precise. We start with the basic necessary definitions.
\begin{definition}
  Given an endofunctor $F \colon \set \rightarrow \set$,
  the category $\DL{F}$ has pairs $(\Sigma, \lambda)$ as objects,
  where $\Sigma \colon \set \rightarrow \set$ is a functor and $\lambda \colon \Sigma F \nattrans F\Sigma$ is a distributive law,
  and a morphism from $(\Sigma_1, \lambda_1)$ to $(\Sigma_2, \lambda_2)$ is a natural transformation $\theta \colon \Sigma_1 \nattrans \Sigma_2$ such that $\lambda_2 \circ \theta F = F\theta \circ \lambda_1$,
  see \cite{KlinN15,LenisaPW00,PowerW02,Watanabe02}.
  The final object $(C,\gamma \colon CF \nattrans FC)$ in $\DL{F}$ is called the \emph{companion} of $F$,
  if it exists. The companion is thus characterised by the property that 
  for every distributive law $\lambda \colon \Sigma F \nattrans F \Sigma$
  there exists a unique natural transformation $\kappa \colon \Sigma \nattrans F$ making the diagram below commute.
  \begin{equation}
    \xymatrix{
      \Sigma F \ar@{=>}[r]^{\kappa F} \ar@{=>}[d]_{\lambda}
      & CF \ar@{=>}[d]^{\gamma}
      \\ F \Sigma \ar@{=>}[r]_{F \kappa}
      & FC
    }
  \end{equation}
\end{definition}
See~\cite{PousR17,BasoldPR17} for a systematic study of the above notion of companion.   
In the following theorem, $\coalg{\Fam{F}}$ denotes the category of $\Fam{F}$-coalgebras
and coalgebra morphisms between them. 

\begin{theorem}[\cite{BasoldPR17}]\label{thm:iso} There exists a functor $\Fam{F} \colon [\set, \set] \rightarrow [\set, \set]$
which gives a one-to-one correspondence between distributive laws $\lambda \colon \Sigma F \nattrans F\Sigma$
and coalgebras $\widehat \lambda \colon \Sigma \nattrans \Fam{F}(\Sigma)$ (natural in $F$). In particular,
there is an isomorphism of categories:
    \begin{equation}\label{eq:iso2}
    \DL{F} \cong \coalg{\Fam{F}} \,.
    \end{equation}
\end{theorem}
In~\cite{BasoldPR17}, $\Fam{F}$ is called the \emph{familiar} of $F$, and is characterised abstractly using
right Kan extensions. The theorem below provides a concrete characterization for the familiar of the functor $FX = A \times X$ and for the isomorphism of Theorem~\ref{thm:iso}. First, $\Fam{F} \colon [\set, \set] \rightarrow [\set, \set]$  is defined for all $\Sigma \colon \set \rightarrow \set$ as
    \begin{equation}\label{def:familiar}
        (A \times \Sigma\--)^{A^{\--}}\colon \set \to \set \,.
    \end{equation}
Second, given a natural transformation of the form $\theta \colon HF \nattrans FG$, we define $\widehat \theta \colon H \nattrans \Fam{F} (G)$,
for all sets $X$, as 
    \begin{equation}\label{eq:unique}
        \xymatrix@C=1.3cm{
            HX \ar[r]^-{H c_X}
            & H(A \times X)^{A^X} \ar[r]^-{\cs^H_{A^X,A \times X}}
            & (H(A \times X))^{A^X} \ar[r]^-{(\theta_X)^{A^X}}
            & (A \times G X)^{A^X}
        }
    \end{equation}
   where the natural transformation $c \colon \Id \nattrans (A \times \--)^{A^{\--}}$ is given for all sets $X$ and $x\in X$ by
    \begin{equation}\label{eq:c-natural}
        c_X(x)(\subs) = (\subs(x), x).
    \end{equation}

\begin{theorem}   \label{thm:streamsfamiliar}
    The familiar of $FX = A \times X$ is $\Fam{F} \colon [\set, \set] \rightarrow [\set, \set]$ defined as in \eqref{def:familiar}. The assignment $\theta\mapsto \widehat \theta$ defined in \eqref{eq:unique}, restricted to the case where $H=G$, extends to a functor which witnesses the isomorphism in \eqref{eq:iso2}.
\end{theorem}

\begin{proof}
In \cite{BasoldPR17}, it is shown that the familiar of an endofunctor $F$ is given as
    \begin{equation*}
        \Fam{F}(G) = \Ran{F}(FG).
    \end{equation*}
    Now we show that, given an endofunctor $G$,
    the right Kan extension $\Ran{F}(FG)$ is given by$(A \times G\--)^{A^{\--}}$,
    with the counit $\counit \colon \Fam{F}(G)F \nattrans FG$ given by
    \begin{equation*}
        \counit_X(a \colon A^{A \times X} \rightarrow A \times G (A \times X))
        = (\id \times G\pi_2)(a(\pi_1))
    \end{equation*}
    on a set $X$, which is easily seen to be natural.
    Thus we have to show that for every $H \colon \set \rightarrow \set$ and
    $\theta \colon HF \nattrans FG$,
    there exists a unique $\widehat\theta \colon H \nattrans \Fam{F}(G)$ making the diagram below commute:
    \begin{equation*}
        \xymatrix{
            {HF} \ar@{=>}[rr]^{\widehat\theta F} \ar@{=>}[dr]_{\theta}
            & {}
            & {\Fam{F}(G)F}
                \ar@{=>}[dl]^{\counit}
            \\ {}
            & {FG}
            & {}
        }
    \end{equation*}

    First we show that the diagram above commutes with the $\widehat\theta$ defined in~(\ref{eq:unique}),
    i.e.\ we show that $\counit \circ \widehat\theta F = \theta$.
    Given $t \in G X$ we have:
    \begingroup
    \allowdisplaybreaks
    \begin{align*}
        & \counit_X(\widehat\theta_{FX}(t)) \\
        & = \{\text{by definition of $\counit$ (counit)}\} \\
        & FG(\pi_2)(\widehat\theta_{FX}(t)(\pi_1)) \\
        & = \{\text{by definition of $\widehat\theta$}\} \\
        & FG(\pi_2)(\theta^{A^{FX}}_{FX}(\cs^H_{A^{FX},FFX}(H(c_{FX})(t)))(\pi_1)) \\
        & = \{\text{by definition of the functor $(\--)^{A^X}$}\} \\
        & FG(\pi_2)(\theta_{FX}(\cs^H_{A^{FX}, FFX}(H(c_{FX})(t))(\pi_1)))) \\
        & = \{\text{by definition of costrength}\} \\
        & FG(\pi_2)(\theta_{FX}(H(\epsilon^{\pi_1})(H(c_{FX})(t)))) \\
        & = \{\text{by naturality of $\theta$}\} \\
        & \theta_X(H(F\pi_2 \circ \epsilon^{\pi_1} \circ c_{FX})(t)) \\
        & = \{\text{as $F\pi_2 \circ \epsilon^{\pi_1} \circ c_{FX} = \id$}\} \\
        & \theta_X(t)
    \end{align*}
    \endgroup
    Next we show that $\widehat\theta$ is the unique natural transformation making the diagram above commute.
    Let $\kappa : H \nattrans \Fam{F}(G)$ such that $\counit \circ \kappa F = \theta$.
    Given $t \in HX$ and $\subs \colon X \rightarrow A$ we have:
    \begingroup
    \allowdisplaybreaks
    \begin{align*}
        & \widehat\theta_X(t)(\subs) \\
        & = \{\text{by definition of $\widehat\theta$}\} \\
        & \theta_X(\cs^H_{A^X, FX}(H(c_X)(t))(\subs)) \\
        & = \{\text{by assumption}\} \\
        & \counit_X(\kappa_{FX}(\cs^H_{A^X, FX}(H(c_X)(t))(\subs))) \\
        & = \{\text{by definition of $\counit$ (counit)}\} \\
        & FG(\pi_2)(\kappa_{FX}(\cs^H_{A^X, FX}(H(c_X)(t))(\subs))(\pi_1)) \\
        & = \{\text{by definition of costrength}\} \\
        & FG(\pi_2)(\kappa_{FX}(H(\epsilon^\subs \circ c_X)(t))(\pi_1))\\
        & = \{\text{by naturality of $\kappa$}\} \\
        & FG(\pi_2)(\Fam{F}(G)(\epsilon^\subs \circ c_X)(\kappa_X(t))(\pi_1)) \\
        & = \{\text{by definition of $\Fam{F}(G)$ on morphisms}\}\\
        & FG(\pi_2)(FG(\epsilon^\subs \circ c_X)(\kappa_X(t)(\pi_1 \circ \epsilon^\subs \circ c_X))) \\
        & = \{\text{as $\pi_2 \circ \epsilon^\subs \circ c_X = \id$ and $\pi_1 \circ \epsilon^\subs \circ c_X = \subs$}\} \\
        & \kappa_X(t)(\subs)
    \end{align*}
    \endgroup
    So $\kappa = \widehat\theta$ and thus $\widehat\theta$ is the natural transformation from $H$
    to $\Fam{F}(G)$ such that $\counit \circ \widehat\theta F = \theta$.
    So $\Fam{F}(G) = \Ran{F}{FG} \cong (A \times G\--)^{A^{\--}}$.
\end{proof}

Now take a monadic stream GSOS specification $\lambda \colon \Sigma F \Rightarrow FT$ and the corresponding distributive law $\rho \colon TF \Rightarrow FT$. Via \eqref{eq:unique}, $\rho$ corresponds to the $\Fam{F}$-coalgebra $\widehat \rho \colon T \Rightarrow (A\times T-)^{A^-}$, which  at the component $\Var$, is a function of the shape $$\widehat \rho_{\Var} \colon T\Var \to (A\times T\Var)^{A^\Var}\text{,}$$ namely, a Mealy machine with input $A^\Var$, output $A$ and state space $T\Var$. 
This is exactly the Mealy machine $m_\lambda$ for open terms defined in~\eqref{eq:mm-ot}, Section~\ref{sec:mm_over_ot}:
\begin{corollary}
With $\rho$ and $\lambda$ as above, we have $\widehat \rho_{\Var} = m_\lambda$.
\end{corollary}
\begin{proof}
	Recall from Section~\ref{sec:mm_over_ot} that $m_\lambda$ is defined as
	\begin{equation}\label{eq:ml1}
\xymatrix{
\openTerms \ar[r]^-{Tc} 
	& T(A \times \Var)^{A^\Var} \ar[r]^-{\bar{\rho}_\Var}
	& (A \times \openTerms)^{A^\Var}\text{.}
}
	\end{equation}
	where $\overline{\rho}$ is the distributive law corresponding to the pointwise extension $\bar{\lambda}$
	of $\lambda$, and $c$ in the above coincides with $c_{\Var}$ in~\eqref{eq:c-natural}. 
	Now, by Lemma~\ref{lm:various-dls}, we have
	\begin{equation}\label{eq:ml2}
	\bar \rho_{\Var} = \left( T(A \times \Var)^{A^{\Var}} \trans[\cs^T_{A^{\Var},A \times \Var}] (T(A \times \Var))^{A^{\Var}} \trans[(\rho_\Var)^{A^{\Var}}] (A \times T\Var)^{A^{\Var}} \right)
	\end{equation}
	Combining~\eqref{eq:ml1} and \eqref{eq:ml2} yields 
	$\widehat \rho_{\Var}$ by definition of $\widehat{\rho}$. Hence 
	$\widehat \rho_{\Var} = m_\lambda$.
\end{proof}

A further observation sheds light on the relationship with the companion.
In~\cite{PousR17}, it is shown that companions of polynomial endofunctors on $\set$ exist,
and a characterisation is given using a generalisation of the notion of a causal function.
For the functor $FX = A \times X$, this notion ends up to be the usual definition of 
causal function, given in Section~\ref{sec:preliminaries}.
To characterise the companion, it is convenient---following Section~\ref{sec:mm_over_ot}---to
say $\cfunc \colon (A^\omega)^X \rightarrow A^\omega$ is causal if the associated $\bar{\cfunc} \colon (A^X)^\omega \rightarrow A^\omega$
(swapping arguments) is causal in the usual sense (Section~\ref{sec:preliminaries}).
Below, we denote the set of all such causal functions by $\Gamma((A^\omega)^X, A^\omega)$. Then, using
the results in \cite{PousR17},
the companion $(C,\gamma \colon CF \nattrans FC)$ of the stream system functor $FX = A \times X$ is given 
by
\begin{equation*}
    CX = \Gamma((A^\omega)^X, A^\omega)\,.
\end{equation*}

Since $(C,\gamma)$ is a final object in $\DL{F}$, there is a unique
natural transformation $\kappa \colon (T,\rho) \nattrans (C,\gamma)$. Its component at $\Var$ is exactly the open semantics from Definition \ref{def:open}, as stated formally by the following result. 

\begin{theorem}\label{thm:semantics-companion}
    Let $\lambda \colon \Sigma F \Rightarrow FT$ be an abstract monadic stream GSOS specification,
    and let $\rho \colon TF \Rightarrow FT$ be the corresponding the distributive law.
    Let $\kappa \colon T \Rightarrow C$ be the unique morphism from $(T,\rho)$ to $(C,\gamma)$.
    Then $\kappa_{\Var} = \osem{\--}$. 
\end{theorem}
We give a self-contained proof of this fact, together with a concrete characterisation of
of the companion (including the associated distributive law), in~\ref{sec:app-companion}.

\section{Conclusions, Related and Future work}
\label{sec:final_remarks}

In this paper we have studied the semantics of open terms specified in the stream GSOS format. 
Our recipe consists in translating the stream specification into a Mealy specification giving semantics to all open terms. Remarkably, this semantics equates two open terms if and only if they are equivalent under all possible interpretations of variables as streams (Theorem~\ref{th:prooftechnique}) or under the interpretation of variables as closed terms (Theorem~\ref{th:closed-terms}). Furthermore, semantic equivalence can be checked by means of the bisimulation proof method enhanced with a technique called up-to substitutions (Theorem~\ref{th:upto}).

Two considerations are now in order. First, the main advantage of using up-to substitutions rather than more standard coinductive techniques (e.g.~\cite{HansenKR16,Klin11}) is that the former may allow finite relations to witness equivalence of open terms (see Example~\ref{ex:uptosubs}), while the latter would often require infinite relations containing all their possible instantiations. We expect this difference to be relevant for equivalence checking algorithms; we leave as future work to properly investigate this issue. 

Second, the correspondences in Theorems~\ref{th:prooftechnique} and~\ref{th:closed-terms} are far from being expected. Indeed, for GSOS specifying labelled transition systems, the notions of bisimilarity of open terms proposed in literature (like the formal hypothesis in~\cite{deSimone85}, the hypothesis preserving in~\cite{Rensink00}, the loose and the strict in~\cite{DBLP:journals/tcs/BaldanBB07}, the rule-matching in~\cite{ACI12}) are sound w.r.t.\ the semantics obtained by the interpretation of variables as closed terms, but \emph{not} complete. Such incompleteness seems to be natural for the authors of~\cite{DBLP:journals/tcs/BaldanBB07} since, they say, in modern open systems, software can be partially specified, executed and then instantiated during its execution. From a coalgebraic perspective, the difficulties in having a complete coinductive characterisation for the semantics of open semantics seems to arises from non-determinism: passing from streams to labelled transition systems means from the coalgebraic outlook to move from the functor $FX=A \times X$ to $FX=\mathcal{P}(A \times X)$.

Endrullis et al.~\cite{EndrullisHB13} consider bisimulation up-to techniques between stream terms, in order to 
improve coinductive methods in Coq. 
Their terms include variables, and up-to substitution is also used there. It is shown
that up-to substitution is sound, in combination with other techniques such as up-to context (for which
causality is assumed; this is equivalent to GSOS-definability~\cite{HansenKR16}). 
The use of variables suggests that open terms may possibly be treated
by the techniques in~\cite{EndrullisHB13} as well; however, bisimulations between
are not mentioned there. Moreover, the approach is different: in particular, it is not based
on the construction via Mealy machines. A more precise understanding of the connection to the current paper is left for future work.

In this sense, our work can be considered as a first concrete step towards a (co)algebraic understanding of the semantics of open terms in the general setting of abstract GSOS~\cite{TP97,Klin11}. Orthogonally to the current paper, there is the more
abstract perspective on distributive laws and abstract GSOS in the recent papers~\cite{BasoldPR17,PousR17}, 
which potentially is of use for bisimilarity of open terms. 
In Section~\ref{sec:familiar}, we connected these two approaches, showing how the current
concrete work on streams fits in the more abstract perspective offered by these papers. 
This is a promising starting point for a general coalgebraic theory
of bisimilarity of open terms, for GSOS specifications for arbitrary functors. 
We leave the development of such a theory for future work.

One of the potential benefits of such a theory might be a general coalgebraic theory of complete axiomatizations. The work of Silva on Kleene coalgebra~\cite{silva2010kleene} is a successful step in this direction, but its scope is limited to regular behaviours. Interestingly enough, one of the first completeness results for regular behaviours~\cite{milner1984complete} already makes use of the semantics of open terms. This is indeed considered in the field to be one of the standard techniques to prove ($\omega$)-completeness for axiomatisations, see the survey in \cite{aceto2005finite}. 

To conclude, it is worth to stress the fact that our approach is confined to GSOS specifications. Extending it to more expressive formats, such as tyft/tyxt~\cite{groote1992structured} or ntyft/ntyxt~\cite{groote1993transition}, seems to be rather challenging. Indeed, while GSOS specification have a neat categorical description in terms of distributive laws \cite{TP97,Klin11}, we are not aware of similar results for more expressive formats. Furthermore, while bisimilarity is guaranteed to be a congruence by these formats,  bisimulations up-to context is in general not compatible: a counterexample for the soundness of up-to context and bisimilarity for the tyft format can be found in \cite{PS11}.



\section*{References}

\bibliographystyle{elsarticle-num}

\newpage

\begin{appendix}

\section{Proof of Theorem~\ref{th:translation}.}
\label{sec:the_enc_is_sound}

\begin{theorem}[Th.~\ref{th:translation}]
Let $(\Sigma, A, R)$ be a stream GSOS specification and $(\tilde \Sigma, A, \tilde R)$ be the corresponding monadic one. Then, 
for all $t \in T_{\Sigma}\emptyset$, $t \bisim \tilde t$.
\end{theorem}

In order to prove Theorem~\ref{th:translation} we have to recall the notion of \emph{disjoint extension}.
\begin{definition}
A (stream) GSOS specification $S' = (\Sigma',A, R')$ is a disjoint extension of
the (stream) GSOS specification $S = (\Sigma, A, R)$ if $\Sigma \subseteq \Sigma'$,
$R \subseteq R'$ and $R'$ adds no new rules for operators in $\Sigma$.
\end{definition}

Because $R'$ adds no new rules for operators in $\Sigma$ we have that the behavior of 
$t \in T_{\Sigma}\emptyset$ is the same for both specifications $S$ and its disjoint extension $S'$.

Let $S = (\Sigma, R, A)$ be a stream GSOS specification and let 
$\tilde S = (\tilde \Sigma, \tilde R, A)$ be its translation to a monadic GSOS specification, see Section~\ref{sec:from_full_to_monadic}.
If we consider the stream GSOS specification $S + \tilde S = (\Sigma \cup \tilde\Sigma, A, R \cup \tilde R)$
we have that $S + \tilde S$ is a disjoint extension of both $S$ and $\tilde S$.
Let $\tilde t$ be the term obtained from $t \in T_{\Sigma}\emptyset$ by replacing each occurence of $f \in \Sigma$ in $t$ 
by $\tilde f \in \tilde \Sigma$.
Taking into account $S + \tilde S$ we will prove that for all $t \in T_{\Sigma}\emptyset$ we have that 
$\tilde t \in T_{\tilde \Sigma}\emptyset$ is such that $t \bisim \tilde t$, 
see Lemma~\ref{lemma:the_enc_is_sound}. 
Because $S + \tilde S$ is a disjoint extension of $\tilde S$, 
we can ensure that Theorem~\ref{th:translation} is sound.

In order to prove the result we define a relation ${\relR} \subseteq \tplus \times \tplus$
and we prove that this relation is a bisimulation up-to bisimilarity. Let ${\relR} \subseteq \tplus \times \tplus$
be the smallest relation satisfying:
\begin{enumerate}[(${\relR}$-i)]
 \item \label{item:relR1}
 $(c, \tilde c) \in {\relR}$ for each constant $c \in \Sigma$.
 \item 
 \label{item:relR2}
 $(t, a.\tilde{t'}) \in {\relR}$ if $t\in T_{\Sigma}\emptyset$ and $t \trans[a] t'$.
 \item \label{item:relR3}
 $(f(t_1, \ldots, t_n), \tilde f(s_1, \ldots, s_n)) \in {\relR}$ whenever $t_i \relR s_i$ for all $i=1, \ldots, n$ and $f\in \Sigma$.
\end{enumerate}

By structural induction on $t$ and considering the different cases based on the definition of 
${\relR}$ we can prove the following results
\begin{lemma}
 \label{lemma:type_of_the_relR}
 If $(t,s) \in {\relR}$ then $t \in T_{\Sigma}\emptyset$ and $s \in T_{\tilde \Sigma} \emptyset$. 
\end{lemma}

\begin{lemma}
 \label{lemma:t_tildet_in_R}
 For all $t \in T_{\Sigma}\emptyset$, $(t,\tilde t) \in {\relR}$.
\end{lemma}
\begin{lemma}
 \label{lemma:R_ext_to_context}
 For all $t(x_1, \ldots, x_m) \in T_{\Sigma}\Var$, 
 $t_1, \ldots, t_m \in T_{\Sigma}\emptyset$, $s_1,\ldots, s_m \in T_{\tilde \Sigma}\emptyset$,
 if $t_i \relR s_i$ for $i=1,\ldots, m$ then $t(t_1, \ldots, t_m) \relR \tilde t(s_1, \ldots, s_m)$.
\end{lemma}

In addition we have:
\begin{lemma}
 \label{lemma:prefix_and_bisimulation}
 For all $t \in \tplus$, if $t \trans[a] t'$ then $t \bisim a.t'$.
\end{lemma}
\begin{proof}
 Suppose that $t_0, t_1, \ldots \in \tplus$ and $a_0, a_1, \ldots \in A $
 are s.t. $t = t_0$, $t' = t_1$, $a_0 = a$ and $t_0\trans[a_0] t_1 \trans[a_1] t_2 \trans[a_2] \ldots$.
 Define ${\relS} = \{(t_i, a_{i}.t_{i+1}) \mid i = 0, 1, \ldots\}$.
 Relation ${\relS}$ is a bisimulation given that for all $i$,
 $t_i \trans[a_i] t_{i+1}$, $a_{i}.t_{i+1}\trans[a_i] a_{i+1}.t_{i+2}$ by rule (\ref{rule:prefix})
 and $t_{i+1} \relS a_{i+1}.t_{i+2}$. 
 Finally $(t, a.t') = (t_0, a_0.t_1) \in {\relS}$ 
\end{proof}

\begin{lemma}
 \label{lemma:R_is_bisim_up_bisim}
 ${\relR}$ is a bisimulation up-to bisimilarity w.r.t. $S + \tilde S$. 
\end{lemma}
\begin{proof}
 We prove ${\relR}$ is a bisimulation up-to bisimilarity by structural induction on $t \in T_{\Sigma}\emptyset$
 where $(t,s) \in {\relR}$ for some $s$.
 Considering $t \in T_{\Sigma}\emptyset$ is enough because of Lemma~\ref{lemma:type_of_the_relR}.
 
 Suppose that $t = c$ is a constant. By Def.~\ref{def:stream_gsos_rule}, $c \trans[a] t'$ with $t'\in T_{\Sigma}\emptyset$
 iff there is axiom $r$, i.e. rule without premises, $c \trans[a] t'$.
 We have two cases to analyze for~$s$ as a consequence of ($\relR$-\ref{item:relR1}) and ($\relR$-\ref{item:relR2}): 
 \begin{inparaenum}[(i)]
 \item Case $s = \tilde c$. By construction of $\tilde S$ there is an axiom $\tilde r$ s.t. 
 $\tilde r =  \tilde c \trans[a] \tilde{t'}$. By Lemma~\ref{lemma:t_tildet_in_R}, $t'\relR \tilde{t'}$.
 \item Case $s = a.\tilde{t'}$. Suppose that $\tilde{t'} \trans[b] s''$ then
 $\tilde{t'} \bisim b.s''$ by Lemma~\ref{lemma:prefix_and_bisimulation}. 
 By (\ref{rule:prefix}), $a.\tilde{t'} \trans[a] b.s''$. 
 Finally, $t' \bisim t' \relR \tilde{t'} \bisim b.s''$ because 
 ${\bisim}$ is reflexive, Lemma~\ref{lemma:t_tildet_in_R} and Lemma~\ref{lemma:prefix_and_bisimulation}. 
 \end{inparaenum}

 For the inductive case we consider $t = f(t_1, \ldots, t_n)$. Suppose
 $f(t_1, \ldots, t_n) \trans[a] t'$. We have two cases to analyze
 as a consequence of ($\relR$-\ref{item:relR2}) and ($\relR$-\ref{item:relR3}).
 Case ($\relR$-\ref{item:relR2}) follows similarly to its counterpart in the base case of the inductive proof.
 For the case~($\relR$-\ref{item:relR3}) we consider 
 $(f(t_1, \ldots, t_n), \tilde f(s_1, \ldots, s_n)) \in {\relR}$ with $t_i \relR s_i$ for all $i=1, \ldots, n$.
 Suppose that 
 $$f(t_1, \ldots, t_n) \trans[a] t(t_1, \ldots t_n, t'_1, \ldots, t'_n)$$
 because of an instantiation of the rule~(\ref{eq:stream_gsos_rule}), then 
 $t_i \trans[a_i] t'_i$ for each $i = 1, \ldots, n$.
 By induction, for each $i$, $t_i \trans[a_i] t'_i$, $s_i \trans[a_i] s'_i$ and there are 
 $\hat t_i$ and $\hat s_i$ s.t. 
 \begin{gather}
 \label{uptobisim:eq:bisim}
 t'_i \bisim \hat t_i \qquad \qquad \qquad  \hat s_i \bisim s'_i \\
  \label{uptobisim:eq:R}
 \hat t_i \relR \hat s_i
 \end{gather}
%
 Using rule~(\ref{rule:general_rule_trans}) and for each $i$, $s_i \trans[a_i] s'_i$, we can derive the following transition 
 $$\tilde f(s_1, \ldots, s_n) \trans[a] \tilde t(a_1.s'_1, \ldots a_n.s_n, s'_1, \ldots, s'_n)$$
 To conclude the proof we have to prove that there are $\hat t$ and $\hat s$ such that 
 $$
 t(t_1, \ldots t_n, t'_1, \ldots, t'_n) \bisim \hat t 
 \relR
  \hat s \bisim \tilde t(a_1.s'_1, \ldots a_n.s'_n, s'_1, \ldots, s'_n)
 $$
 To prove this, notice first that for each $i$
\begin{gather}
 \label{uptobisim:eq:bisim:ind-case}
 t_i \bisim a_i.t'_i \bisim a_i.\hat t_i 
 \qquad \qquad \qquad \qquad 
 a_i.\hat s_i \bisim a_i.s'_i \\
 \label{uptobisim:eq:R:ind-case}
 \qquad a_i.\hat t_i \relR a_i.\hat s_i
\end{gather}
 The left side equation of (\ref{uptobisim:eq:bisim:ind-case}) is a consequence of
 Lemma~\ref{lemma:prefix_and_bisimulation} and (\ref{uptobisim:eq:bisim}), taking into account 
 that $\bisim$ is a congruence for the prefix operators.
 The right side equation of (\ref{uptobisim:eq:bisim:ind-case}) is also a consequence of this last fact.
 Equation (\ref{uptobisim:eq:R:ind-case}) is a consequence of ($\relR$-\ref{item:relR3}) and (\ref{uptobisim:eq:R}).

 Define $\hat t = t(a_1.\hat t_1, \ldots, a_n.\hat t_n, \hat t_1, \ldots,\hat t_n)$ and 
 $\hat s  = \tilde t(a_1.\hat s_1, \ldots, a_n.\hat s_n, \hat s_1, \ldots,\hat s_n)$.
 Recall that $\bisim$ is a congruence for all operator defined using the stream GSOS specification format and this property can be extended to arbitrary context constructed using these operators.
 Taking into account this fact and (\ref{uptobisim:eq:bisim}) and (\ref{uptobisim:eq:bisim:ind-case})  
 we get $t(t_1, \ldots t_n, t'_1, \ldots, t'_n) \bisim \hat t$ and $\hat s \bisim \tilde t(a_1.s'_1, \ldots a_n.s_n, s'_1, \ldots, s'_n)$.  Finally $\hat t \relR \hat s$ because of Lemma~\ref{lemma:R_ext_to_context}, (\ref{uptobisim:eq:R}) and (\ref{uptobisim:eq:R:ind-case}).
\end{proof}

\begin{lemma}
 \label{lemma:the_enc_is_sound}
 Let $S$ and $\tilde S$ be, resp., a stream GSOS specification and its encoding in monadic GSOS specification.
 For the GSOS specification $S+\tilde S$ we have $t\bisim \tilde t$ for all $t \in T_{\Sigma}\emptyset$.
\end{lemma}
\begin{proof}
 The result is a straight consequence of Lemmas~\ref{lemma:t_tildet_in_R} and~\ref{lemma:R_is_bisim_up_bisim}
\end{proof}

\section{Strength and Costrength}\label{sec:strengthco}

The categorical notions of strength and costrength are recalled here, as they
play an important role from Section~\ref{sec:pe_monadic_gsos} onwards. 
%
The \emph{strength} of an endofunctor $F$ on $\set$ is a map 
$\st^F_{A,X} \colon A \times FX \to F(A \times X)$ natural in $A$ and $X$, defined by
$\st^F_{A,X}(a,t) = (F\eta^a)(t)$ where $\eta^a : X \to A \times X$ is given by $\eta^a(x) = (a, x)$.
The \emph{costrength} of $F$ is a map
$\cs^F_{A,X} \colon F (X^A) \to (FX)^A$ natural in $A$ and $X$, defined by
$\cs^F_{A,X}(t)(a) = (F\epsilon^a)(t)$
where $\epsilon^a : X^A \to X$ is given by $\epsilon^a(f) = f(a)$.

We recall several basic properties that will only be used in proofs, and can be safely skipped
by the reader. First of all, the following diagrams commute:
\begin{equation}
\label{eq:strength}
\begin{tabular}{cc}
\begin{tikzpicture}
  \matrix (m) [matrix of math nodes,row sep=2em,column sep=4em]
  {
     A \times \Sigma V & \Sigma (A \times V) \\
     \Sigma V &   \\};
  \path[-stealth]
    (m-1-1) edge node [left] {\scriptsize $\pi_2$} (m-2-1)
            edge node [above] {\scriptsize $\st^\Sigma_{A,V}$} (m-1-2)
    (m-1-2) edge node [below right] {\scriptsize $\Sigma \pi_2$} (m-2-1);
\end{tikzpicture}
& 
\begin{tikzpicture}
  \matrix (m) [matrix of math nodes,row sep=2em,column sep=4em]
  {
     1 \times \Sigma V & \Sigma (1 \times V) \\
     \Sigma V &   \\};
  \path[-stealth]
    (m-1-1) edge [-] node [left] {\scriptsize $\simeq$} (m-2-1)
            edge node [above] {\scriptsize $\st^\Sigma_{1,V}$} (m-1-2)
    (m-1-2) edge [-] node [below right] {\scriptsize $\Sigma \simeq$} (m-2-1);
\end{tikzpicture}
\end{tabular}
\end{equation}
For the triangle on the left-hand side, we have
$
\Sigma \pi_2 \circ \st^\Sigma_{A,V}(a,t) = 
\Sigma \pi_2 \circ \Sigma\eta^a(t) = \Sigma (\pi_2 \circ \eta^a)(t) =
t = 
\pi_2(a,t)
$. 
The triangle on the right-hand side is a special case. 

We will also use that costrength is natural in the endofunctor involved, i.e., for any natural transformation
$\gamma \colon F \Rightarrow G$ and any sets $A,X$, the following diagram commutes:
\begin{equation}\label{eq:cs-nat}
\begin{tabular}{c}
 \begin{tikzpicture}
  \matrix (m) [matrix of math nodes,row sep=2.5em,column sep=4em]
  {
    F(X^A)&  (FX)^A \\
    G(X^A) & (GX)^A \\
   };     
    \path[-stealth]
      (m-1-1) edge node [above] {\scriptsize $\cs^F_{A,X}$} (m-1-2)
              edge node [left] {\scriptsize $\gamma_{X^A}$} (m-2-1)
      (m-2-1) edge node [above] {\scriptsize $\cs^G_{A,X}$} (m-2-2)
      (m-1-2) edge node [right] {\scriptsize $(\gamma_X)^A$} (m-2-2)
     ;
\end{tikzpicture}
\end{tabular}
\end{equation} 
An analogous property holds for strength, see~\cite[Appendix A]{HK11}.

Finally, when $T$ is a monad, the costrength $\cs^T_{A,-}$ is always a distributive law of 
the monad $T$ over the functor $(-)^A$ (Diagram~\eqref{eq:dl-monad} with $F = (-)^A$), see, e.g.,~\cite[Example 2]{Jacobs06}.

\section{The Pointwise Extension}
\label{app:pointwise_ext}

We recall the general definition of pointwise extension.
This requires a few preliminary notions, in particular of a map $\ev$; the
reader can safely skip these general definitions and move 
immediately to the instance where the functor at hand is the one for streams, 
for which we give a concrete characterisation.
\begin{definition}
 Let $s : Y \to B \times Y$ be an $(B \times {-})$-coalgebra. 
 For an $F^B$-coalgebra $m : X \to (FX)^B$, define the $F$-coalgebra
 $$s \ltimes m : Y \times X \to F (Y \times X)$$
 to be the function composition:
 \begin{align*}
 Y \times X 
 & \trans[s \times m] (B\times Y) \times (FX)^B \trans[\simeq] Y \times (B \times (FX)^B) \\
 & \trans[\id_Y \times \epsilon_{FX}] Y \times FX \trans[\st^F_{Y,X}] F(Y\times X)
 \end{align*}
\end{definition}

Let $\tuple{Z, \zeta}$ be a final $F$-coalgebra. 
By finality, for all $s : Y \to B \times Y$ and $m : X \to (FX)^B$,
there is a unique $F$-coalgebra morphism $\sem{-}_{s \ltimes m} : Y \times X \to Z$:
\begin{equation}
\label{eq:beh_s_ltimes_m}
\begin{tabular}{c}
\begin{tikzpicture}
  \matrix (m) [matrix of math nodes,row sep=2em,column sep=6em]
  {
     Y \times X & Z \\
     F(Y \times X) & FZ \\
};     
   \path[-stealth]
     (m-1-1) edge [dotted] node [above] {\scriptsize $\sem{-}_{s \ltimes m}$}  (m-1-2)
             edge node [left] {\scriptsize $s \ltimes m$}  (m-2-1)
     (m-2-1) edge node [above] {\scriptsize $F\sem{-}_{s \ltimes m}$} (m-2-2)
     (m-1-2) edge node [right] {\scriptsize $\zeta$}  (m-2-2); 
\end{tikzpicture}
\end{tabular}
\end{equation}


We instantiate the above by taking $s$ to
be the final ($B \times {-}$)-coalgebra $(B^\omega, \tuple{\hd,\tl})$, and
$m$ to be the final $F^B$-coalgebra $(\bar{Z},\bar{\zeta})$; $\ev$ is then defined
to be the coinductive extension below. 

\begin{equation}
\label{eq:ev_diagram}
\begin{tabular}{c}
\begin{tikzpicture}
  \matrix (m) [matrix of math nodes,row sep=2em,column sep=6em]
  {
     B^\omega \times \bar{Z} & Z \\
     F(B^\omega \times \bar{Z}) & FZ\\
};     
   \path[-stealth]
     (m-1-1) edge [dotted] node [above] {\scriptsize $\ev$}  (m-1-2)
             edge node [left] {\scriptsize $\tuple{\hd,\tl} \ltimes \bar \zeta$}  (m-2-1)
     (m-2-1) edge node [above] {\scriptsize $F\ev$} (m-2-2)
     (m-1-2) edge node [right] {\scriptsize $\zeta$}  (m-2-2); 
\end{tikzpicture}
\end{tabular}
\end{equation}
In the current paper,
we are interested in the case that $Z = A^\omega$ consists of streams (the final coalgebra of stream systems),
and $\bar{Z} = \Gamma(B^\omega,A^\omega)$ consists of causal functions (the final coalgebra of Mealy machines). 
In that case, $\ev \colon B^\omega \times \Gamma(B^\omega, A^\omega) \rightarrow A^\omega$ is given concretely by
\begin{equation}\label{eq:ev}
 \ev(\alpha, f) = f(\alpha) \,.
\end{equation}
as reported in~\cite[Example 4.3]{HK11}. 


\begin{definition}\label{def:pw-ext}
 For a signature $\Sigma$, an algebra $\bar \sigma : \Sigma \bar Z \to \bar Z$ is a pointwise
 extension of $\sigma : \Sigma  Z \to Z$ if the following diagram commutes:
 \begin{equation}
\label{eq:cat_point_ext}
\begin{tabular}{c}
\begin{tikzpicture}
  \matrix (m) [matrix of math nodes,row sep=3em,column sep=4.5em]
  {
     B^\omega \times \Sigma \bar Z &  \Sigma(B^\omega \times \bar Z) & \Sigma (Z) \\
     B^\omega \times \bar Z & & Z\\
};     
   \path[-stealth]
     (m-1-1) edge node [above] {\scriptsize $\st^\Sigma_{B^\omega, \bar Z}$}  (m-1-2)
             edge node [left] {\scriptsize $\id\times \bar \sigma$}  (m-2-1)
     (m-2-1) edge node [above] {\scriptsize $\ev$} (m-2-3)
     (m-1-2) edge node [above] {\scriptsize $\Sigma\ev$}  (m-1-3)
     (m-1-3) edge node [right] {\scriptsize $\sigma$}  (m-2-3);
\end{tikzpicture}
\end{tabular}
\end{equation} 
\end{definition}
This generalises the concrete instance for streams in Definition~\ref{def:pw-ext-concrete}.

\subsection{Proof of Theorem~\ref{th:pointwise_ext}}
\label{sec:proof_of_pointwise_ext}

We need the following lemma to prove Theorem~\ref{th:pointwise_ext}:
\begin{lemma}\label{lm:various-dls}
	Let $\lambda \colon \Sigma F \nattrans FT$  be a monadic abstract GSOS specification,
	with pointwise extension $\bar{\lambda} \colon \Sigma F^B \Rightarrow (FT)^B$. Let
	$\bar{\rho} \colon T F^B \Rightarrow (FT)^B$ be the distributive law 
	induced by $\bar{\lambda}$. Then
$$
\bar \rho = \left( T(FX)^B \trans[\cs^T_{B,FX}] (TFX)^B \trans[(\rho_X)^B] (FTX)^B \right)
$$
where $\rho \colon TF \Rightarrow FT$ is the distributive law induced by $\lambda$.
\end{lemma}
\begin{proof}
	Let $\bar{\rho}' =  (T(FX)^B \trans[\cs^T_{B,FX}] (TFX)^B \trans[(\rho_X)^B] (FTX)^B)$.
	By~\cite[Lemma 3.4.27]{Bartels04}, $\bar{\rho}$ is the unique distributive law satisfying
	$\bar{\lambda} = \bar{\rho} \circ \kappa_{F^B} \circ \Sigma \eta_{F^B}$ (this is an instance
	of the one-to-one correspondence between (monadic) GSOS specifications and distributive laws).
	 Hence,
	it suffices to show the following properties:
	\begin{enumerate}
		\item $\bar{\rho}'$ is a distributive law of monad over functor;
		\item 
	$\bar{\lambda} = \bar{\rho}' \circ \kappa_{F^B} \circ \Sigma \eta_{F^B}$. 
	\end{enumerate}
	For 1., consider the following diagram. 
	$$
	\xymatrix@C=1.3cm{
		TT(FX)^B \ar[ddd]_{\mu_{(FX)^B}} \ar[r]^{T\cs^T_{B,FX}} 
			& T(TFX)^B \ar[d]^{\cs^T_{B,TFX}} \ar[r]^{T(\rho_X)^B}
			& T(FTX)^B \ar[d]^{\cs^T_{B,FTX}} \\
			& (TTFX)^B \ar[dd]^{(\mu_{FX})^B} \ar[r]^{(T\rho_X)^B}
			& (TFTX)^B \ar[d]^{(\rho_{TX})^B}\\
			& & (FTTX)^B \ar[d]^{(F\mu_X)^B}\\
		T(FX)^B \ar[r]^{\cs^T_{B,FX}} 
			& (TFX)^B \ar[r]^{(\rho_X)^B}
			& (FTX)^B\\
		(FX)^B \ar[u]^{\eta_{(FX)^B}}
		\ar[ur]_{(\eta_{FX})^B}
		\ar@(r,d)[rru]_{(F\eta_X)^B}
	}
	$$
	The left part of the diagram commutes since $\cs^T_{B,-}$ is a distributive law of the monad $T$
	over the functor $(-)^B$, the right part follows since $\rho$ is a distributive law
	of the monad $T$ over the functor $F$ (and the top-right square by naturality of $\cs^T_{B,-}$).

	For 2., consider the following diagram.
	  \begin{equation}
	  \xymatrix@C=1.3cm@R=1.3cm{
	    \Sigma (FX)^B \ar@(u,u)[rr]^{\bar{\lambda}_X} \ar[r]^{\cs^{\Sigma}_{B,FX}} \ar[d]_{{\nu}_{(FX)^B}}
	    	& (\Sigma FX)^B \ar[r]^{(\lambda_X)^B} \ar[d]_{({\nu}_{FX})^B}
	    	& (FTX)^B\\
	   T(FX)^B  \ar[r]^{\cs^T_{B,FX}} \ar@/_40pt/[rru]_{\bar{\rho}'_X}
	   & (TFX)^B \ar[ur]^{(\rho_X)^B}
	   & \\	  
	  }
\end{equation} 
where $\nu = \kappa \circ \Sigma \eta$. The square commutes by~\eqref{eq:cs-nat}, the triangle on
to the right since $\rho$ is the distributive law extending $\lambda$, and the upper and lower shapes
by definition of $\bar{\lambda}$ and $\bar{\rho}'$ respectively.
\end{proof}


This allows us to prove Theorem~\ref{th:pointwise_ext}.
\begin{proof}[Proof of Theorem~\ref{th:pointwise_ext}]
Let  $\asem{-}^{\bar{\rho}} \colon T\bar{Z} \rightarrow \bar{Z}$ 
and $\asem{-}^\rho \colon TZ \rightarrow Z$ be the algebras
arising by finality from $\bar{\rho}$ and $\rho$ respectively, as
defined in diagram (B) at the end of Section~\ref{sec:spec}. 

Let $\nu = \kappa \circ \Sigma \eta$.
The algebras $\sigma\colon \Sigma Z \rightarrow Z$ and $\bar{\sigma} \colon \Sigma \bar{Z} \rightarrow \bar{Z}$ are defined respectively by
$\asem{-}^{\rho} \circ \nu_{Z}$
and
$\asem{-}^{\bar{\rho}} \circ \nu_{\bar{Z}}$. We need to prove
that the latter is a pointwise extension of the former.

By Lemma~\ref{lm:various-dls}, the
distributive law 
$\bar{\rho}$
induced by $\bar{\lambda}$ is given by
$$\bar \rho = (T(FX)^B \trans[\cs^T_{B,X}] (TFX)^B \trans[(\rho_X)^B] (FTX)^B)\,.$$
Hence, by~\cite[Theorem 6.1]{HK11}, $\asem{-}^{\bar{\rho}}$ is a pointwise
extension of $\asem{-}^{\rho}$, i.e., the lower rectangle in the diagram below commutes:
 \begin{equation}
\begin{tabular}{c}
\begin{tikzpicture}
 \node at (-1.6cm, 0.8cm) {\scriptsize $(i)$ };
 \node at (2.2cm, 0.8cm) {\scriptsize nat. };
 \node at (0.6cm, -0.8cm) {\scriptsize $\asem{-}^{\bar{\rho}}$ is a p.e. of $\asem{-}^\rho$};
 \matrix (m) [matrix of math nodes,row sep=2.5em,column sep=4em]
 {
   B^\omega \times \bar Z & \Sigma(B^\omega \times \bar Z)& \Sigma Z\\
   B^\omega \times T\bar Z & T(B^\omega \times \bar Z) & TZ  \\
   B^\omega \times \bar Z && Z \\
  };     
   \path[-stealth]
        (m-1-1) edge node [above] {\scriptsize $\st^\Sigma_{B^\omega, \bar Z}$ } (m-1-2)
                edge node [left] {\scriptsize $\id \times \nu_{\bar Z}$} (m-2-1)
        (m-2-1) edge node [above] {\scriptsize $\st^F_{B^\omega, \bar Z}$ } (m-2-2)
                edge node [left] {\scriptsize $\id \times \asem{-}^{\bar{\rho}}$} (m-3-1)
        (m-3-1) edge node [below] {\scriptsize $\ev$ } (m-3-3)
        (m-1-2) edge node [above] {\scriptsize $\Sigma\ev$ } (m-1-3)
                edge node [left] {\scriptsize $\nu_{B^\omega \times \bar Z}$} (m-2-2)
        (m-2-2) edge node [above] {\scriptsize $T\ev$ } (m-2-3)
        (m-1-3) edge node [right] {\scriptsize $\nu_Z$ } (m-2-3)
        (m-2-3) edge node [right] {\scriptsize $\asem{-}^\rho$ } (m-3-3)
       ;
\end{tikzpicture}
\end{tabular}
\end{equation} 
The square $(i)$ is naturality of strength (see Section~\ref{sec:strengthco}), and the square on the
right commutes by naturality of $\nu$.
\end{proof}

\section{Non-compatibility for non-monadic specifications}\label{app:is_not_compatible}

In our approach, we first transform an arbitrary stream GSOS specification into a monadic one and then we construct the corresponding Mealy specification.
One could transform directly an arbitrary specification into a Mealy specification following the approach in \cite{HK11}, but Theorem~\ref{th:upto} in Section \ref{sec:compatibility} would not go through. In this appendix we shortly explain why.

\medskip

The pointwise extension of GSOS specification that is not monadic 
requires the introduction of a family of auxiliary operators, the so-called \emph{buffer operators}.
For a stream GSOS specification $\lambda \colon \Sigma(\Id \times F) \Longrightarrow FT$ with $F = A \times X$, 
there is a Mealy GSOS specification  $\bar \lambda \colon \bar \Sigma (\Id \times (F)^{A^\Var}) \Longrightarrow (F\bar T)^{A^\Var}$  
that pointwise extends $\lambda$ (see \cite[Theorem 6.3]{HK11}). 
Notice that that the signature used for the Mealy machine is not the original $\Sigma$ but $\bar \Sigma$. 
In our setting, this signature is the extension of $\Sigma$ with the buffer operators ${\subs} \buffer$, 
for each $\subs \colon \Var \to A$. The extended specification $\bar \lambda$ defines their semantics as follows.

\begin{equation*}
 \label{eq:buffers}
 \dedrule{x \trans[\subs|a] x'}{\subs \buffer x \trans[\subs'|a] \subs' \buffer x'}
\end{equation*}

In addition, for each rule $r$ with shape \eqref{eq:stream_gsos_rule} of the stream GSOS specification and $\subs : \Var \to A$, 
the Mealy machines has a rule $r'$ defined by 
 \begin{equation*}
 r' = \dedrule{x_1 \trans[\subs | a_1] x'_1 \quad \cdots \quad x_n \trans[\subs | a_n]x'_n}
            {f(x_1, \ldots, x_n) \trans[\subs | a] t(\subs \buffer x_1, \ldots, \subs \buffer x_n,x'_1, \ldots, x'_n)} 
 \end{equation*}

By following this approach, open bisimilarity is \emph{not} preserved by substitution.
%
%
\medskip 

For example, consider  stream systems over reals and the family of ``bad'' prefix operators $a.x$,  for each $a\in \R$. 
The following left and right rules represent respectively, a stream GSOS specification and its corresponding Mealy specification. 
\begin{equation*}
\label{eq:prefix}
\dedrule{}{a.x \trans[a] x} 
\qquad \qquad \qquad
\dedrule{}{a.x \trans[\subs | a] \subs \buffer x} 	
\end{equation*}

For a fixed $a \in \R$, let $\hat \subs : \Var \to \R$ be such that $\hat\subs(\vx) = a$ and $\hat\subs(\vy) \neq 0$. 
Then $a.\vx$ and $\hat \subs \buffer \vx$ are in $\obisim$ because

\begin{center}$a.\vx \trans[\subs | a] \subs \buffer \vx$ and
$\hat \subs \buffer \vx \trans[\subs | \hat\subs(\vx)] \subs \buffer \vx$. 
\end{center}
Consider now substitution $\theta: \Var \to \openTerms$ such that $\theta (\vx) = \vx \oplus \vy$, where $\oplus$ is defined as in Figure \ref{fig:rulestreamcalculus}(a). 
Then 
$\theta(a.\vx)$ and $\theta(\hat \subs \buffer \vx)$ are not bisimilar because
$\theta(a.\vx) = a.(\vx \oplus \vy) \trans [\subs|a] \subs \buffer (\vx \oplus \vy)$,
$\theta(\hat \subs \buffer \vx) = \hat \subs \buffer (\vx \oplus \vy) %
 \trans[\subs | \hat\subs(\vx) + \hat\subs(\vy)] \hat \subs \buffer(\vx \oplus \vy)$
and $a \neq \hat\subs(\vx) + \hat\subs(\vy)$.

\medskip
This fact also entails that up-to substitutions is \emph{not} compatible. Indeed, following the general theory in \cite{BonchiPPR17}, if it would be compatible, then open bisimilarity would be closed under substitution.

\section{Details on the companion}\label{sec:app-companion}

In this section, we precisely characterise the companion $(C,\gamma)$ of the functor $F$, $F(X) = A \times X$,
show how this companion yields final Mealy machines, and ultimately prove Theorem~\ref{thm:semantics-companion}.
As already mentioned in Section~\ref{sec:familiar}, $C$ is given by:
\begin{equation}\label{eq:comp-functor}
    CX = \Gamma((A^\omega)^X, A^\omega)\,.
\end{equation}
We will also prove below that $\gamma \colon CF \Rightarrow FC$ is given by:
    \begin{equation} \label{eq:charcomp}
        \gamma_X(\cfunc \colon (A^\omega)^{A \times X} \to A^\omega) =
        (\head(\cfunc((a,x) \mapsto a.\alpha)), \cfunc')
    \end{equation}
    where $\alpha \in A^\omega$ is an arbitrary stream,
    and $\cfunc' \colon (A^\omega)^\Var \to A^\omega$ is given by
    \begin{equation}
        \cfunc'(\psi \colon X \to A^\omega) = \tail(\cfunc((a,x) \mapsto a.\psi(x))).
    \end{equation}
In the proof that $(C,\gamma)$ is the companion, we will use that the $(C,\gamma)$
`constructs final Mealy machines' (Proposition~\ref{prop:finalmealy2}). First,
we recall the following characterisation of final Mealy machines for $F^{A^{\Var}}$ 
(which is well-known; see, e.g.,~\cite{DBLP:journals/cuza/HansenR10}).
\begin{lemma} \label{lemma:finalmealy}
  The final $(A \times \--)^{A^\Var}$-coalgebra $\zeta \colon \tgamma \to (A \times \tgamma)^{A^\Var}$
  is given by
  \begin{equation}
      \zeta(\cfunc \colon (A^\omega)^\Var \to A^\omega)(\subs \colon \Var \to A)
      = (\head(\cfunc(\vx \mapsto \subs(\vx).\alpha)), \cfunc_\subs)
  \end{equation}
  where $\alpha \in A^\omega$ is an arbitrary stream,
  and $\cfunc_\subs \colon (A^\omega)^\Var \to A^\omega$ is given by
  \begin{equation}
      \cfunc_\subs(\psi \colon \Var \to A^\omega) = \tail(\cfunc(\vx \mapsto \subs(\vx).\psi(\vx))).
  \end{equation}
\end{lemma}

\begin{proposition} \label{prop:finalmealy2}
    The $X$ component of the $\Fam{F}$-coalgebra $\widehat\gamma$,
    corresponding to the distributive law $\gamma$ given in (\ref{eq:charcomp}),
    is the final $(A \times \--)^{A^X}$-coalgebra.
\end{proposition}
\begin{proof}
    This follows from a simple computation.
    \begingroup
    \allowdisplaybreaks
    \begin{align*}
        & \widehat\gamma_X(\cfunc \colon (A^\omega)^X \to A^\omega)(\subs \colon X \to A) \\
        & = \{\text{by definition of $\widehat\gamma$}\} \\
        & \gamma_X^{A^X}(\cs^C_{A^X,A \times X}(C(c_X)(\cfunc)))(\subs)\\
        & = \{\text{by definition of the functor $(\--)^{A^X}$}\} \\
        & \gamma_X(\cs^C_{A^X,A \times X}(C(c_X)(\cfunc))(\subs)) \\
        & = \{\text{by definition of costrength}\} \\
        & \gamma_X(C(\epsilon^\subs \circ c_X)(\cfunc)) \\
        & = \{\text{by definition of $\epsilon$}\} \\
        & \gamma_X(C(x \mapsto (\subs(x), x))(\cfunc)) \\
        & = \{\text{by definition of the functor $C$ on morphisms}\} \\
        & \gamma_X((\varphi \colon A \times X \to A^\omega) \mapsto \cfunc(x \mapsto (\subs(x), x)))
    \end{align*}
    \endgroup
    But after spelling out $\gamma_X$,
    it is immediately clear that this is exactly the Mealy machine given in Lemma~\ref{lemma:finalmealy}. 
\end{proof}

Now, we can prove that $(C,\gamma)$ as presented above is really the companion.

\begin{proposition} \label{prop:charcomp}
    The companion of $FX = A \times X$ is given by $(C,\gamma)$ as
    defined in~\eqref{eq:comp-functor} and \eqref{eq:charcomp} in the beginning of this section.

    Moreover, given a distributive law $\rho \colon TF \nattrans FT$,
    the natural transformation $\kappa \colon T \nattrans C$,
    which is the unique morphism from $(\Sigma, \lambda)$ to the companion $(C, \gamma)$,
    is given by
    \begin{equation}
        \kappa_X(t \in TX)(\psi \colon \Var \to A^\omega)
        = \asem{(T\psi)(t)} \,.
    \end{equation}
\end{proposition}
\begin{proof} 
    It is not very difficult to verify that $\kappa_X(t)$ is a causal function.

    In order to show that $(C, \gamma)$ is the companion,
    we first show that $\kappa$ is a morphism from $(T, \rho)$ to $(C, \gamma)$,
    i.e.\ that $F\kappa \circ \rho = \gamma \circ \kappa F$.
    First observe that
    \begin{equation*}
        \gamma(\kappa_{FX}(t \in TFX))
        = (\hd(\asem{T((a,x) \mapsto a.\alpha)(t)}), \cfunc'),
    \end{equation*}
    where $\cfunc' \colon (A^\omega)^X \to A^\omega$ is given by
    \begin{equation*}
        \cfunc'(\psi \colon X \to A^\omega)
        = \tl(\asem{T((a,x) \mapsto a.\psi(x))(t)}).
    \end{equation*}
    Now for the left component we have
    \begingroup
    \allowdisplaybreaks
    \begin{align*}
        & \hd(\asem{T((a,x) \mapsto a.\alpha)(t)}) \\
        & = \{\text{as $\hd = \pi_1 \circ \zeta$}\} \\
        & \pi_1(\zeta(\asem{T((a,x) \mapsto a.\alpha)(t)})) \\
        & = \{\text{by definition of $\asem{\--}$}\} \\
        & \pi_1(F(\asem{\--})(\rho_{A^\omega}(T(\zeta)(T((a,x) \mapsto a.\alpha)(t))))) \\
        & = \{\text{using $\pi_1 \circ Ff = \pi_1$}\} \\
        & \pi_1(\rho_{A^\omega}(T(\zeta)(T((a,x) \mapsto a.\alpha)(t)))) \\
        & = \{\text{as $\zeta(a.\alpha) = (a, \alpha)$}\} \\
        & \pi_1(\rho_{A^\omega}(T((a,x) \mapsto (a, \alpha))(t))) \\
        & = \{\text{by definition of the functor $F$ on morphisms}\} \\
        & \pi_1(\rho_{A^\omega}(TF(x \mapsto \alpha)(t))) \\
        & = \{\text{by naturality of $\rho$}\} \\
        & \pi_1(FT(x \mapsto \alpha)(\rho_X(t))) \\
        & = \{\text{using $\pi_1 \circ Ff = \pi_1$}\} \\
        & \pi_1(\rho_X(t))
    \end{align*}
    \endgroup
    and for the right component we have
    \begingroup
    \allowdisplaybreaks
    \begin{align*}
        & \cfunc'(\psi \colon X \to A^\omega) \\
        & = \{\text{by definition of $\cfunc'$}\} \\
        & \tl(\asem{T((a,x) \mapsto a.\psi(x))(t)}) \\
        & = \{\text{as $\tl = \pi_2 \circ \zeta$}\} \\
        & \pi_2(\zeta((\asem{T((a,x) \mapsto a.\psi(x))(t)}))) \\
        & = \{\text{by definition of $\asem{\--}$}\} \\
        & \pi_2(F(\asem{\--})(\rho_{A^\omega}(T(\zeta)(T((a,x) \mapsto a.\psi(x))(t))))) \\
        & = \{\text{as $\zeta(a.\psi(x)) = (a, \psi(x))$}\} \\
        & \pi_2(F(\asem{\--})(\rho_{A^\omega}(T((a,x) \mapsto (a, \psi(x)))(t)))) \\
        & = \{\text{by definition of the functor $F$ on morphisms}\} \\
        & \pi_2(F(\asem{\--})(\rho_{A^\omega}(TF(\psi)(t)))) \\
        & = \{\text{by naturality of $\rho$}\} \\
        & \pi_2(F(\asem{\--})((FT(\psi)(\rho_X(t))))) \\
        & = \{\text{using $\pi_2 \circ Ff = f \circ \pi_2$}\} \\
        & \asem{T(\psi)(\pi_2(\rho_X(t)))} \\
        & = \{\text{by definition of $\kappa$}\} \\
        & \kappa_X(\pi_2(\rho_X(t)))(\psi)
    \end{align*}
    \endgroup
    Thus, combining both calculations, we get $F(\kappa_X)(\rho_X(t)) = \gamma_X(\kappa_{FX}(t))$.

    Finally we show that $\kappa$ is also the unique such morphism.
    Let $\theta$ be a morphism from $(T, \rho)$ to $(C, \gamma)$.
    Then both $\theta$ and $\kappa$ are coalgebra morphisms from $(T, \widehat\rho)$ to $(C, \widehat\gamma)$.
    And for a set $X$, $\theta_X$ and $\kappa_X$ are coalgebra morphisms from $(TX, \widehat\rho_X)$ to $(CX, \widehat\gamma_X)$,
    but $(CX, \widehat\gamma_X)$ is the final Mealy machine according to Proposition~\ref{prop:finalmealy2},
    so $\kappa_X = \theta_X$ for all sets $X$ and thus $\kappa = \theta$.
\end{proof}
Putting everything together, we obtain the proof of Theorem~\ref{thm:semantics-companion}.

%
\begin{proof}[Proof of Theorem~\ref{thm:semantics-companion}]
    For any $t \in \openTerms$ and $\psi \colon \Var \rightarrow A^\omega$:
    \begin{equation*}
        \kappa_\Var(t)(\psi) = \asem{(T\psi)(t)} = \osem{t}(\psi)
    \end{equation*}
    by Proposition~\ref{prop:charcomp} and Proposition~\ref{prop:osem_and_csem},
    and thus $\kappa_\Var = \osem{\--}$.
\end{proof}

\end{appendix}

\end{document}